%% file: main.tex
\documentclass[runningheads]{llncs}

\usepackage{eccv}

\usepackage{eccvabbrv}
\usepackage{enumitem}

\usepackage{mathtools}
\usepackage{graphicx}
\usepackage{booktabs}
\usepackage{multirow}
\usepackage{colortbl} 
\usepackage{placeins} 

\usepackage[accsupp]{axessibility}  

\usepackage[breaklinks,colorlinks,citecolor=eccvblue]{hyperref}
\usepackage[pass]{geometry} 

\usepackage{orcidlink}

\definecolor{oursblue}{RGB}{222,240,255}
\definecolor{baselinegray}{RGB}{245,245,245}
\definecolor{timegray}{RGB}{230,230,230}

\newcommand{\method}{CubicSplat\xspace}
\newcommand{\BezierSplat}{B\'ezier Splatting\xspace}
\newcommand{\DiffVG}{DiffVG\xspace}
\newcommand{\LIVE}{LIVE\xspace}
\newcommand{\LIVSS}{LIVSS\xspace}
\newcommand{\ThreeDGS}{3D Gaussian Splatting\xspace}

\AtBeginDocument{%
  \crefname{figure}{Fig.}{Figs.}
  \Crefname{figure}{Figure}{Figures}
  \crefname{equation}{Eq.}{Eqs.}
  \Crefname{equation}{Equation}{Equations}
  \crefname{subsection}{Sec.}{Secs.}
  \Crefname{subsection}{Section}{Sections}
  \crefname{subsubsection}{Sec.}{Secs.}
  \Crefname{subsubsection}{Section}{Sections}
  \crefname{theorem}{Theorem}{Theorems}
  \Crefname{theorem}{Theorem}{Theorems}
  \crefname{lemma}{Lemma}{Lemmas}
  \Crefname{lemma}{Lemma}{Lemmas}
  \crefname{proposition}{Proposition}{Propositions}
  \Crefname{proposition}{Proposition}{Propositions}
  \crefname{corollary}{Corollary}{Corollaries}
  \Crefname{corollary}{Corollary}{Corollaries}
  \crefname{definition}{Definition}{Definitions}
  \Crefname{definition}{Definition}{Definitions}
  \crefname{remark}{Remark}{Remarks}
  \Crefname{remark}{Remark}{Remarks}
}

\usepackage{graphicx}

\begin{document}

\title{CubicSplat: Differentiable Vector Graphics via Error-Bounded Forward Relaxation}

\titlerunning{CubicSplat}

\author{Chenglong Liu\inst{1}\orcidlink{0009-0009-2422-7529} \and
Xin Zhang\inst{1}\orcidlink{0009-0006-8048-7091} \and
Yimeng Zhu\inst{1}\orcidlink{0009-0002-2231-7650} \and
Liyang He\inst{1}\orcidlink{0000-0002-1609-0747} \and
Yixiao Ma\inst{1}\orcidlink{0009-0002-7849-6950} \and
Yu Su\inst{2,3}\orcidlink{0000-0002-7950-4919} \and
Zhenya Huang\inst{1}\orcidlink{0000-0003-1661-0420} \and
Qi Liu\inst{1}\orcidlink{0000-0001-6956-5550}\thanks{Corresponding author: qiliuql@ustc.edu.cn.}}

\authorrunning{C.~Liu et al.}

\institute{State Key Laboratory of Cognitive Intelligence, University of Science and Technology of China, Hefei, 230088, China
\email{\{sepcnt,zx2020,zzyymm,heliyang,mayx\}@mail.ustc.edu.cn}\\
\email{\{huangzhy,qiliuql\}@ustc.edu.cn}\\ \and
Hefei Normal University, Hefei, 230061, China\\
\email{yusu@hfnu.edu.cn}\\ \and
Zhejiang Key Laboratory of Intelligent Education Technology and Application, Zhejiang Normal University, Jinhua, 321004, Zhejiang, China
}

\maketitle

\input{sec/0_abstract}
\input{sec/1_intro}
\input{sec/2_related}
\input{sec/3_method}
\input{sec/4_experiments}
\input{sec/5_conclusion}

\section*{Acknowledgements}
This work was supported by grants from the National Natural Science Foundation of China (62525606, U25B2072) and Open Research Fund of Zhejiang Key Laboratory of Intelligent Education Technology and Application (No. 2025ZNJYKF006).

%
%
\bibliographystyle{splncs04}
\bibliography{main}

\clearpage
\newgeometry{margin=1in}
\appendix

\setcounter{section}{0}
\renewcommand{\thesection}{S\arabic{section}}
\setcounter{equation}{0}
\renewcommand{\theequation}{S\arabic{equation}}
\setcounter{figure}{0}
\renewcommand{\thefigure}{S\arabic{figure}}
\setcounter{table}{0}
\renewcommand{\thetable}{S\arabic{table}}

\renewcommand{\theHsection}{supp.\arabic{section}}
\renewcommand{\theHsubsection}{supp.\arabic{section}.\arabic{subsection}}
\renewcommand{\theHsubsubsection}{supp.\arabic{section}.\arabic{subsection}.\arabic{subsubsection}}
\renewcommand{\theHfigure}{supp.\arabic{figure}}
\renewcommand{\theHtable}{supp.\arabic{table}}
\renewcommand{\theHequation}{supp.\arabic{equation}}

\begin{center}
{\LARGE\textbf{Supplementary Materials}}\\[0.5em]
{\large CubicSplat: Differentiable Vector Graphics via Error-Bounded Forward Relaxation}
\end{center}

\input{sec/6_appendix}
\FloatBarrier

\end{document}

%% file: sec/0_abstract.tex
\begin{abstract}
Vector graphics are prized for their resolution independence, compact storage, and direct editability, making differentiable optimization of their parametric primitives an attractive goal. Yet classical rasterization is discontinuous with respect to geometry, and existing remedies that smooth the forward pass demand increasingly elaborate heuristics as scene complexity grows. We trace this fragility to a \textbf{gradient seesaw}: design choices that improve forward geometric exactness can systematically degrade the induced gradient signal, and vice versa. To navigate this tension we introduce \textbf{CubicSplat}, a differentiable vector rasterizer that replaces B\'{e}zier closest-point solvers with uniform polyline surrogates whose geometric error is bounded at $\mathcal{O}(S^{-2})$. The resulting static computation graph yields well-conditioned gradients by construction, while a compositing-derived visibility mechanism prunes degenerate primitives without auxiliary regularization. On DIV2K and Kodak benchmarks CubicSplat achieves state-of-the-art reconstruction quality with over \textbf{2\,dB} PSNR gain in the closed-fill setting, while training up to $\textbf{4}{\times}$ faster than prior methods. The code is available at \url{https://github.com/CubicSplat/repo}

\keywords{Differentiable rendering \and Vector graphics \and Scene representation \and Gradient conditioning}
\end{abstract}

%% file: sec/1_intro.tex
\section{Introduction}
\label{sec:intro}

Vector graphics (VGs) encode images as resolution-independent, compact, and directly editable parametric primitives. Fitting these parameters by gradient descent is conceptually elegant but hindered by discontinuous rasterization, which yields gradients that are either zero almost everywhere or explosive near boundaries. A natural remedy is to redesign the forward pass so that pixels respond smoothly to geometric change.

In 3D mesh rendering, SoftRas~\cite{softras} trades a softened forward image for smooth gradients, while nvdiffrast~\cite{nvdiffrast} preserves a crisp forward pass and recovers boundary gradients via deferred antialiasing. DiffVG~\cite{Li:2020:DVG} brings the same idea to 2D vector graphics, replacing hard inside/outside decisions with area-averaged coverage. Follow-up methods such as LIVE~\cite{xu2022live}, LIVSS~\cite{livss}, SGLIVE~\cite{sglive}, and B\'{e}zier Splatting~\cite{liu2025bezier} have achieved impressive results, yet a persistent frustration remains: as scenes grow in complexity, these methods require opacity heuristics, auxiliary regularization losses, and ad-hoc annealing schedules to keep optimization from collapsing. The sheer volume of such scaffolding suggests that something more fundamental is at play.

A simple experiment exposes the missing piece (\cref{sec:ablation}). Increasing forward fidelity (finer sampling, adaptive subdivision, tighter B\'{e}zier approximation) monotonically reduces geometric error,
yet the quality of the learned representation plateaus almost
immediately while runtime grows and optimization can become less stable. This plateau is hard to reconcile with approximation error as the bottleneck, but follows naturally from the dual role of the forward pass: it must be accurate enough to evaluate the objective, yet more critically, it must induce gradients well-conditioned enough to navigate
the objective landscape.

We identify this tension as the \textbf{gradient seesaw}: 
design choices that improve forward exactness can systematically degrade the quality of the induced gradient signal, while an overly coarse forward approximation can itself become the binding source of error. 
Three concrete mechanisms illustrate the effect: \textbf{(i)}~\textbf{Exact forward, ill-conditioned backward.}
Differentiating through iterative closest-point or root-finding~\cite{Li:2020:DVG} procedures for exact B\'{e}zier distances introduces ill-conditioning at degenerate configurations (Appendix A.2). The forward model is geometrically faithful, yet the gradients it induces become unreliable precisely where optimization needs them most.
\textbf{(ii)}~\textbf{Smooth gradients, bloated forward.}
Point-sample based methods~\cite{liu2025bezier} that place Gaussian splats along curves yield well-behaved gradients, yet the discrete summation can deviate structurally from the true geometry around high-curvature regions and does not commute with spatial rescaling (Appendix A.3). When the learned vector representation is rendered at resolutions higher than those used during training, the fixed sampling budget cannot keep pace, and the reconstructed image progressively diverges from the intended geometry. \textbf{(iii)}~\textbf{Adaptive refinement, discontinuous graph.}
Adaptive subdivision reduces forward geometric error on demand, yet each discrete split event changes the computation graph, producing non-smooth, heteroscedastic gradient variance across primitives that undermines consistent optimization.

The upshot is that fidelity operates in a regime of diminishing, and eventually negative, returns: a more exact forward model can be a worse gradient oracle, yet a forward model that is too inexact optimizes the wrong objective.

Motivated by this analysis, we propose \textbf{Error-bounded Forward Relaxation}: we intentionally relax the forward rasterization to obtain a stable, well-conditioned gradient oracle while bounding the induced forward bias. We instantiate this principle in \textbf{CubicSplat}, a tile-parallel differentiable vector graphics renderer that achieves throughput comparable to B\'{e}zier Splatting without coupling geometric fidelity or multi-scale reconstruction quality to a per-primitive sampling budget. The core idea is to approximate each B\'{e}zier curve with a uniform polyline surrogate so that all distance queries reduce to closed-form point-to-segment computations and no iterative solver enters the backward pass. This eliminates the dominant source of gradient ill-conditioning while introducing only a bounded, controllable geometric bias. CubicSplat further supports closed shapes through winding-number fill and handles occlusion via standard Porter--Duff compositing, providing a unified treatment of open strokes and filled regions without auxiliary sampling or ad-hoc regularization. As a practical byproduct, the compositing pass naturally yields per-primitive visibility statistics that we exploit to maintain efficiency at large primitive counts.

Our experiments validate these claims. On a 200-image DIV2K subset with $N{=}1024$ primitives, CubicSplat achieves $25.78$\,dB (closed) / $25.93$\,dB (open) versus B\'{e}zier Splatting's $23.45$ / $25.45$\,dB, while reducing wall-clock training time by $4{\times}$ (closed) and $1.8{\times}$ (open). The method is also more parameter-efficient: with only $256$ curves CubicSplat already matches the PSNR that DiffVG requires $1024$ curves to reach. A single forward--backward step on a $2040{\times}1344$ image with 2048 curves is $1.4{\times}$ faster than B\'{e}zier Splatting, $120{\times}$ faster than DiffVG. Beyond these core results, rendering remains resolution-faithful up to $24\mathrm{K}$ and primitive counts scale to $32\mathrm{K}$ on a single consumer GPU. Ablations further confirm sampling-density invariance, reinforcing that the polyline surrogate acts as a better-conditioned optimization oracle rather than a crude approximation.

Our contributions are as follows:
\begin{itemize}
  \item We identify the \textbf{forward-fidelity vs.\ gradient-quality tension} intrinsic to differentiable vector rendering and provide empirical evidence through a broad sampling-density invariance regime.
  \item We propose \textbf{Error-bounded Forward Relaxation}, bringing the classical polyline approximation into a principled oracle-design framework: a family $R_S = R \circ \Pi_S$ of surrogate renderers with $\mathcal{O}(S^{-2})$ geometric bias, paired with compositing-derived visibility pruning that replaces explicit geometric regularization with renderer-intrinsic capacity allocation.
  \item We demonstrate that \textbf{CubicSplat} achieves state-of-the-art quality--efficiency tradeoffs, substantially outperforming strong baselines on both open and closed primitives while requiring a fraction of the training time.
\end{itemize}

%% file: sec/2_related.tex
\section{Related Work}
\label{sec:related}

\subsection{Differentiable Rasterization}
Early differentiable renderers recover gradients through boundary integrals or Monte Carlo edge sampling~\cite{loper2014opendr,kato2018neural3d}. SoftRas~\cite{softras} introduces probabilistic rasterization that yields smooth gradients at the cost of a softened forward image, whereas nvdiffrast~\cite{nvdiffrast} preserves a crisp forward pass and recovers boundary gradients via deferred antialiasing.
Approximate differentiable renderers~\cite{keselman2022algebraic} deliberately trade forward-image fidelity for smooth gradients, and Wang et al.~\cite{sdf_simple} explicitly analyze the bias-variance tradeoff governed by a smoothing parameter in SDF rendering. For splatting-style rasterizers, unbiased gradient estimation under stochastic sampling has also been studied~\cite{mueller2022poisson}.

\subsection{Differentiable Vector Graphics}
\DiffVG~\cite{Li:2020:DVG} introduced a fully differentiable rasterizer for vector graphics with both an analytical prefiltering mode and a Monte Carlo antialiasing mode. The analytical path differentiates through iterative cubic B\'ezier closest-point solvers, whose ill-conditioning we analyze in Appendix A.2; the Monte Carlo mode avoids this issue but incurs substantial per-pixel sampling cost, making it prohibitive for high-resolution scenes.
\LIVE~\cite{xu2022live}, \LIVSS~\cite{livss}, and SGLIVE~\cite{sglive} reduce inter-curve gradient interference by incrementally adding paths. On the representation side, Im2Vec~\cite{reddy2021im2vec} learns
differentiable path representations via closed-form B\'ezier integrals, and CLIPasso~\cite{clipasso} optimizes stroke sequences under semantic losses.

\BezierSplat~\cite{liu2025bezier} accelerates differentiable vector rendering by distributing 2D Gaussian primitives along B\'ezier curves. Because image formation relies on discrete point samples\cite{zhang2024gaussianimage,zhang2024imagegscontentadaptiveimagerepresentation}, the method inherits phase- and density-dependent aliasing that requires denser sampling to suppress, and its discretized summation does not commute with spatial rescaling, coupling optimization outcomes to sampling hyperparameters (Appendix A.3). Our method sidesteps these issues with an explicit, error-bounded polyline surrogate that delivers stable gradients without tying approximation quality to a per-primitive sampling budget.

\subsection{Gaussian Splatting and Gradient Quality}
\ThreeDGS~\cite{kerbl3Dgaussians} and its variants~\cite{scaffoldgs,Huang2DGS2024,10.1007/978-3-031-73016-0_23,guedon2023sugar,guo2024tetsphere} achieve real-time novel-view synthesis through differentiable Gaussian rasterization. Recent studies expose several gradient pathologies in this family of methods: approximate opacity compositing versus volumetric rendering~\cite{celarek2025}, intra-primitive gradient collision~\cite{absgs}, opacity bias from densification rules~\cite{revisedensification}, and forward/backward depth-order mismatch~\cite{stopthepop}. These findings reinforce the broader theme that forward approximations can have outsized effects on gradient quality.

%% file: sec/3_method.tex
\section{Method}

We view differentiable vector graphics learning as gradient-oracle design.
Rather than pushing the forward renderer toward maximal geometric exactness,
we construct a family of relaxed forward passes that yield stable,
well-conditioned gradients while keeping the forward bias bounded and
interpretable.

\begin{figure}[t]
  \centering
  \includegraphics[width=\linewidth]{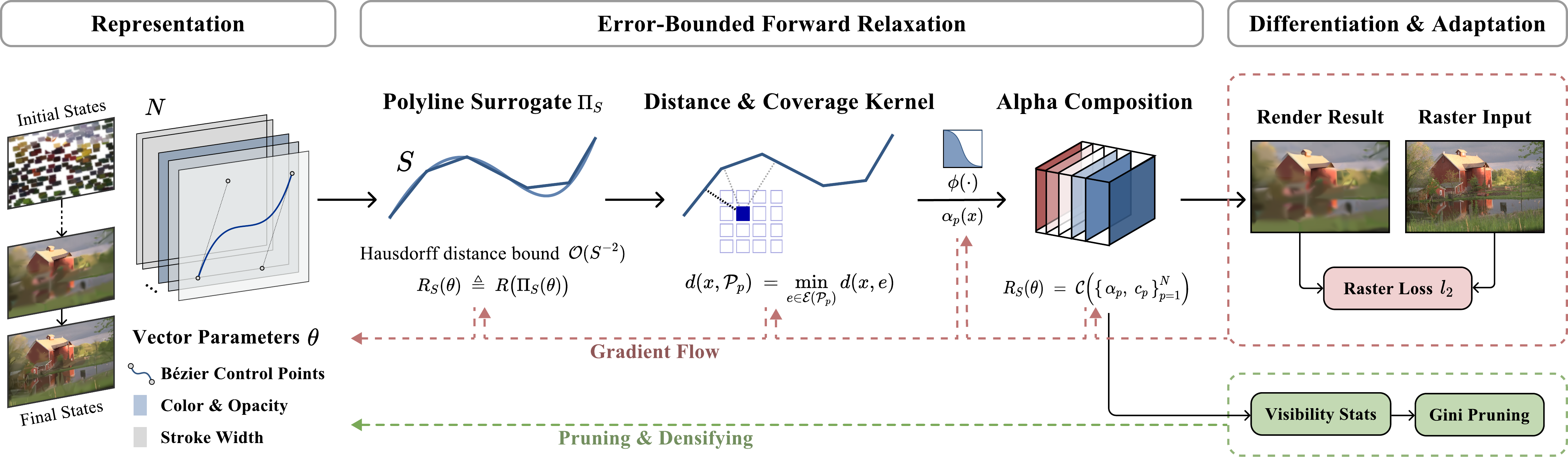}
  \caption{Overview of the \textbf{CubicSplat} pipeline. 1) Each primitive
is encoded as a set of learnable B\'{e}zier parameters. 2) The forward pass applies the
error-bounded relaxation: a polyline surrogate $\Pi_S$ converts
B\'{e}zier geometry into closed-form distance queries, and a smooth
coverage kernel $\phi$ maps the signed residual to per-pixel opacity,
which is assembled into the final image via alpha compositing. 3) The raster loss drives gradient updates back to the vector parameters, while compositing-derived visibility
statistics feed into Gini pruning and densification.}
  \label{fig:structure}
\end{figure}

\subsection{The Framework of CubicSplat}
\label{sec:overview}

Let $\theta$ collect the continuous parameters of $N$ layered vector
primitives: control points, thickness $r_p$, color $c_p$, and opacity.
The front-to-back layer order is fixed throughout optimization.
Given a target image $I^*$, we solve
\begin{equation}
  \min_{\theta}\;\mathcal{L}\!\bigl(R_S(\theta),\,I^*\bigr),
\end{equation}
where $R_S$ is a relaxed differentiable rasterizer parameterized by
a sampling density $S$, and $\mathcal{L}$ is a pixel-space loss.

The renderer $R_S$ determines the gradient oracle
$\nabla_\theta\mathcal{L}$; in vector graphics, many common failures
(unstable convergence, phantom gradients under occlusion, sensitivity to
sampling rules) are better explained by oracle conditioning than by forward
image fidelity alone. We therefore design $R_S$ to optimize oracle quality
subject to a bounded forward bias, rather than to maximize geometric
exactness.

Concretely, $R_S$ decomposes into three stages that map parameters to
pixels:
\begin{equation}
\label{eq:pipeline}
  \alpha_p(x)
  \;=\;
  \underbrace{
    \phi\!\left(
      \frac{\delta_p(x,\Pi_S(\mathcal{B}_p))}{\tau_p}
    \right)
  }_{\text{coverage kernel (\cref{sec:coverage})}},
  \quad
  R_S(\theta)
  \;=\;
  \underbrace{
    \mathcal{C}\!\left(
      \left\{\,\alpha_p,\;c_p\,\right\}_{p=1}^{N}
    \right)
  }_{\text{compositing (\cref{sec:adapt})}},
\end{equation}
where $\delta_p(x)$ is a \textbf{topology-dependent signed residual}
(\cref{sec:coverage}) that encodes how far pixel $x$ lies inside or
outside primitive $p$, and:
\begin{enumerate}
\item $\Pi_S(\mathcal{B}_p)$ is a \textbf{polyline surrogate}
  (\cref{sec:surrogate}) that approximates each B\'ezier primitive with a
  uniform $S$-segment polyline, reducing all distance queries to
  closed-form point-to-segment computations and eliminating iterative
  solvers from the backward pass;
\item $\phi\!\bigl(\cdot\bigr)$ is a smooth, monotone \textbf{coverage
  kernel} (\cref{sec:coverage}) that maps the signed residual to soft
  per-pixel opacity via an anti-aliasing band of width $\tau_p$;
\item $\mathcal{C}$ is a front-to-back \textbf{alpha-compositing} pass
  (\cref{sec:adapt}) whose intermediate transmittance values yield
  per-primitive visibility statistics, enabling inequality-triggered
  pruning via renderer-derived signals rather than explicit geometric
  regularization.
\end{enumerate}

The sampling density $S$ enters only through the surrogate $\Pi_S$:
increasing $S$ monotonically reduces geometric approximation error at
$\mathcal{O}(S^{-2})$, but our ablations reveal a broad
solution-equivariance regime beyond a modest threshold $S_0$, confirming
that $S$ functions as a relaxation knob for oracle conditioning rather
than an approximation budget (\cref{sec:ablation}).

The remainder of this section details each stage. Only the training
oracle is relaxed; at evaluation time the learned parameters can be
rendered with an arbitrarily high $S$ or an exact B\'ezier renderer.

\subsection{Error-Bounded Forward Relaxation}
\label{sec:surrogate}

Consider a B\'ezier segment $\mathcal{B}(t)$, $t\in[0,1]$. We define a
surrogate operator $\Pi_S$ that maps each B\'ezier segment to a uniform
$S$-segment polyline:
\begin{equation}
\mathcal{P}_S = \Pi_S(\mathcal{B}) = \{v_0,\ldots,v_S\},\qquad
v_i = \mathcal{B}\!\left(\frac{i}{S}\right),
\end{equation}
and similarly for an entire primitive consisting of one or more B\'ezier
segments. Any fixed uniform rule fits the same framework; we adopt a
fixed rule to preserve a static computation graph.

This induces a family of relaxed renderers:
\begin{equation}
R_S(\theta)\ \triangleq\ R\big(\Pi_S(\theta)\big),
\end{equation}
where $\Pi_S(\theta)$ applies the polyline surrogate to all B\'ezier
geometry encoded in $\theta$.

Increasing $S$ monotonically improves geometric fidelity: $\Pi_S$
converges to the original B\'ezier geometry as $S\to\infty$. We bound
the geometric discrepancy between $\mathcal{B}$ and $\Pi_S(\mathcal{B})$~\cite{rockafellarwets1998, clarke1990, farin2002curves}
(e.g., a Hausdorff-distance bound decaying as $\mathcal{O}(S^{-2})$
under mild smoothness conditions), which implies that the induced
coverage bias is also controlled. This establishes $R_S$ as an
error-controlled relaxation of $R$; a formal treatment is given in
Appendix A.4.

While $\Pi_S$ becomes geometrically exact as $S$ increases, our
ablations reveal a broad solution-equivariance regime: there exists a
threshold $S_0$ such that for all $S\ge S_0$, optimizing
$\min_{\theta} \mathcal{L}(R_S(\theta), I^*)$ yields learned
representations $\theta_S$ that are geometrically equivalent and nearly
identical in reconstruction quality, even though the forward surrogate
continues to grow more accurate.

This behavior is inconsistent with an approximation-error-dominated
narrative. If optimization were primarily limited by forward
approximation error, increasing $S$ would systematically alter the
learned solution. Instead, beyond $S_0$, $S$ chiefly governs oracle
conditioning and runtime. We therefore interpret $S$ as a
forward-relaxation knob that trades forward fidelity for a
better-conditioned gradient oracle. Only the training oracle is relaxed:
the learned vector parameters can always be rendered with a higher $S$
or an exact B\'ezier renderer at evaluation time (\cref{sec:scaling}).

\subsection{Differentiable Coverage from Signed Residual}
\label{sec:coverage}

We define coverage through a smooth kernel applied to a signed residual
that measures how far a pixel lies inside or outside a primitive.

Let $\mathcal{P}_p$ denote the polyline surrogate for primitive $p$. The
Euclidean distance from a pixel center $x$ to the polyline is
\begin{equation}
d(x,\mathcal{P}_p)\ =\ \min_{e\in\mathcal{E}(\mathcal{P}_p)} d(x,e),
\end{equation}
where $\mathcal{E}(\mathcal{P}_p)$ is the set of line segments and
$d(x,e)$ is the closed-form point-to-segment distance. This avoids
iterative closest-point solvers and yields a stable computation graph.

We construct a signed residual $\delta_p(x)$ that is positive inside the
primitive and negative outside, encoding both geometry and topology in a
single continuous quantity:
\begin{equation}
\label{eq:signed_residual}
\delta_p(x, \Pi_S(\mathcal{B}_p)) \;=\;
\begin{cases}
  r_p \;-\; d(x,\,\mathcal{P}_p),
    & \text{open stroke},\\[4pt]
  s_p(x)\;\cdot\; d(x,\,\mathcal{P}_p),
    & \text{closed fill},
\end{cases}
\end{equation}
where $r_p\ge 0$ is the learnable stroke half-width, and
$s_p(x)\in\{-1,+1\}$ is the sign derived from a nonzero winding-number
test on the polyline boundary ($s_p=+1$ if $x$ is inside,
$s_p=-1$ otherwise). For open strokes, a pixel is inside when it
falls within distance $r_p$ of the curve; for closed fills, inside/outside
is determined entirely by the winding rule consistent with the nonzero fill convention.

The coverage kernel $\phi$ in \cref{eq:pipeline} maps the normalized
signed residual $\delta_p/\tau_p$ to soft opacity through a smooth,
monotone function $\phi:\mathbb{R}\to[0,1]$. The anti-aliasing band
$\tau_p>0$ controls the transition width: $\phi$ saturates to $1$ well
inside the primitive and decays to $0$ well outside, with the steepest
response at the boundary ($\delta_p=0$) where optimization most needs
informative gradients. The specific kernel form (half-sided Gaussian for
strokes, logistic sigmoid for fills) and its motivation are detailed in
Appendix A.5.

The winding-number sign $s_p(x)$ is a discrete quantity; we treat it as
constant with respect to $\theta$ and route all gradients through the
continuous distance $d(x,\mathcal{P}_p)$. This is justified because
$s_p$ changes only when the polyline boundary crosses the pixel, a
codimension-1 event in parameter space
(Appendix A.4). At the boundary itself, the kernel
$\phi$ already provides the relevant gradients for that transition. This
design is consistent with the observed $S$-equivariance in closed-mode
optimization: once the polyline boundary reaches sufficient fidelity,
the winding-number fill is already correct for nearly all pixels, and
learning is driven primarily by boundary motion under the prefiltered
band.

\subsection{Compositing and Visibility-Driven Adaptation}
\label{sec:adapt}

We render primitives in front-to-back order via standard alpha compositing
with initial transmittance $T_0(x)=1$ and per-primitive color $c_i$:
\begin{equation}
C(x)=\sum_{i=1}^N T_{i-1}(x)\,\alpha_i(x)\,c_i,\qquad
T_i(x)=T_{i-1}(x)\,(1-\alpha_i(x)).
\end{equation}

A key design principle is to avoid explicit geometric constraints on the
parametric curves (e.g.\ convexity or non-self-intersection penalties),
which would introduce nontrivial auxiliary losses and complicate the
optimization landscape. Instead, we allow arbitrary curve configurations
to emerge freely during training and detect degeneracy where it matters
by inspecting inexpensive, renderer-derived signals. Rasterization
already exposes per-pixel transmittance $T_{i-1}(x)$ and per-primitive
contribution $T_{i-1}(x)\,\alpha_i(x)$ as intermediate quantities. We
aggregate these into a per-primitive importance score:
\begin{equation}
w_p \;\propto\; \sum_x T_{p-1}(x)\,\alpha_p(x),
\end{equation}
optionally combined with a visibility rate and a geometric coverage
estimate. All terms are renderer-derived and available at negligible
overhead.

Rather than regularizing geometry, we monitor how concentrated importance
is across primitives. When a small subset dominates, as quantified by a
dispersion metric such as the Gini coefficient over $\{w_p\}$, we prune
primitives whose $w_p$ remains persistently small, directly removing
degenerate or fully occluded elements from the representation. Freed
capacity can optionally be reallocated by reinitializing new primitives near regions of high residual error, following a strategy
similar to LIVE~\cite{xu2022live}.

%% file: sec/4_experiments.tex
\section{Experiments}
\label{sec:experiments}

We evaluate \method on image vectorization, focusing on
(i)~reconstruction quality for a fixed primitive budget and
(ii)~efficiency in both single forward--backward pass and end-to-end training time.

\subsection{Datasets and Implementation Details}

We adopt the large-scale evaluation protocol introduced by B\'{e}zier Splatting~\cite{liu2025bezier}, which moved beyond the few low-resolution test images common in earlier work to diverse, high-resolution benchmarks. Specifically, we evaluate on a 200-image subset of DIV2K~\cite{Timofte_2017_CVPR_Workshops} and the full 24-image Kodak dataset~\cite{kodak1999}, all at original resolution. We report PSNR, SSIM, and LPIPS~\cite{zhang2018perceptual}, and compare against DiffVG~\cite{Li:2020:DVG}, LIVE~\cite{xu2022live}, LIVSS~\cite{livss}, and B\'{e}zier Splatting~\cite{liu2025bezier}. Following the standard DiffVG convention, we parameterize each closed primitive as two connected cubic B\'{e}zier segments and each open primitive as three connected cubic Bézier segments, using 10 control points in both cases. We optimize control points, color, width, and opacity with
Adan~\cite{adan} under an $\ell_2$ loss, using learning rates of $0.1$
for color and opacity, $10^{-3}$ for control points, and $10^{-4}$ for
width. We train for 12{,}000 steps in open mode and 10{,}000 steps in
closed mode without early stopping; all baselines use their respective
default configurations. The default polyline sampling density is
$S{=}24$ for open curves and $S{=}12$ for closed curves; the
quality--speed trade-off across $S$ is analyzed in \cref{sec:ablation}. Primitives are pruned and immediately densified
every 500 steps when the Gini coefficient exceeds 0.35, with pruning
disabled for the final 2{,}000 steps. At evaluation time, we render with $S^*{=}24$ or the training density,
whichever is larger, ensuring that the forward surrogate is at least as
accurate as during optimization.

\input{tables/speed}
\begin{figure}[tp]
  \centering
  \includegraphics[width=0.95\linewidth]{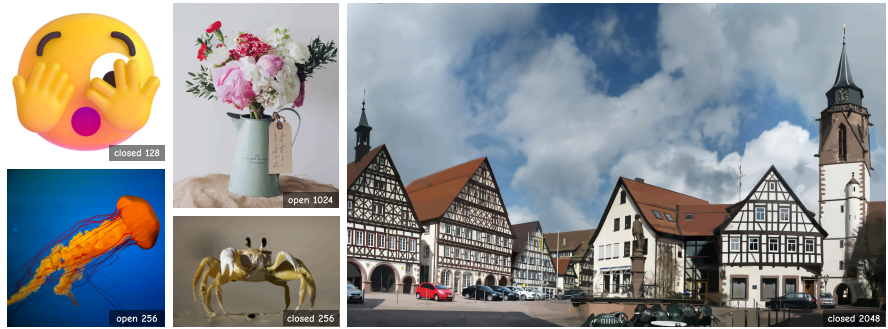}
  \caption{Qualitative results across varying budgets, topology modes, and datasets~\cite{Timofte_2017_CVPR_Workshops,kodak1999,fluent}. At low budgets (closed 128, open/closed 256), CubicSplat preserves clean boundaries and smooth color gradients, producing coherent stylization rather than noisy artifacts. At higher budgets (open 1024, closed 2048), fine details such as the half-timbered fa\c{c}ade emerge faithfully, confirming that the gradient oracle scales from compact stylization to high-fidelity reconstruction.}
  \label{fig:showcase}
\end{figure}

\subsection{Single Training-step Speed}
\label{sec:speed}
\Cref{tab:speed} reports single training-step times on a
$2040{\times}1344$ image with 2048 curves at default sampling densities.
On open curves, \method completes a forward--backward pass in 6.77\,ms,
$\textbf{1.4}{\times}$ faster than \BezierSplat and over
$\textbf{120}{\times}$ faster than \DiffVG. On closed curves the gap
over \BezierSplat widens to roughly $\textbf{2}{\times}$, while \DiffVG
remains over an order of magnitude slower. The speed advantage stems from representation efficiency: B\'ezier
Splatting discretizes each curve into densely sampled Gaussian
primitives that can reach millions in closed-fill settings, whereas
CubicSplat requires only a modest number of polyline segments per curve
with closed-form distance queries, avoiding both dense splat
accumulation and iterative solvers.

\subsection{Quantitative Results}
\label{sec:quantitative}

\input{tables/comparison}
\begin{figure}[!t]
  \centering
  \includegraphics[width=0.95\linewidth]{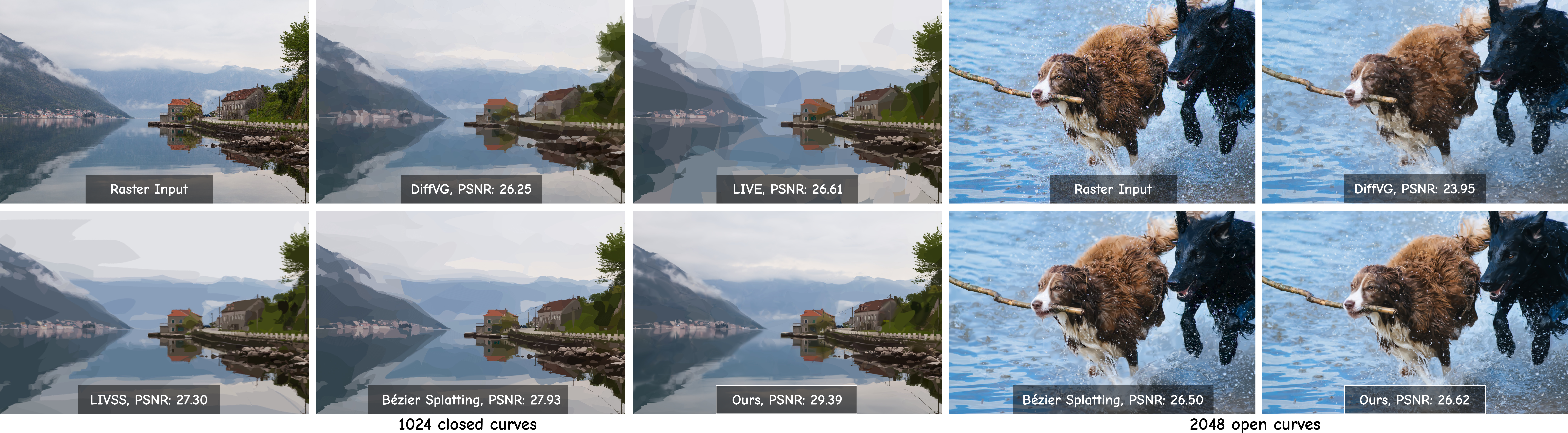}
  \caption{Qualitative comparison with state-of-the-art differentiable VG rasterization methods, including DiffVG~\cite{Li:2020:DVG}, LIVE~\cite{xu2022live}, LIVSS~\cite{livss}, and B\'{e}zier Splatting~\cite{liu2025bezier}. The lake scene illustrates gradient and reflection reconstruction, while the dog example highlights preservation of fine details such as eye color and facial structure.}
  \label{fig:lake_dog}
\end{figure}

\Cref{tab:comparison} and \Cref{tab:comparisons_kodak} summarize reconstruction quality and training cost for
various primitive counts under both modes.
\textbf{\method establishes new state-of-the-art} PSNR and SSIM while
substantially reducing end-to-end training time.
The gains are consistent in the open setting and become decisive for closed shapes,
where accurate filling and stable gradients are critical.

\textbf{\method is also more parameter-efficient}: visibility-driven pruning removes redundant and
occluded curves, concentrating capacity on primitives that actually contribute to the image.
As a result, \method with $256$ curves already matches the PSNR of DiffVG with
$1024$ curves in both modes (\cref{tab:comparison}).

The improvements are most pronounced in the closed mode.
At $N{=}1024$, \method delivers more than 2\,dB PSNR improvement over \BezierSplat while
cutting per-image training time by roughly $4{\times}$, demonstrating that the
quality gains accompany rather than trade against efficiency.
In the open mode, \method still improves PSNR and SSIM at every budget and reduces training
time markedly, yielding a strictly better quality--time trade-off.

We observe that B\'ezier Splatting can achieve lower LPIPS for open strokes despite worse PSNR and SSIM. We view this as a difference in failure mode, not fidelity. CubicSplat's error-bounded surrogate yields clean, piecewise-smooth coverage with sharp, scale-consistent boundaries (\cref{fig:detail}); at low budgets this produces faithful but stylized results that under-represent photographic micro-texture. Sample-based splatting instead introduces stochastic high-frequency fluctuations that can incidentally mimic fine texture (Appendix A.3), which LPIPS, itself trained on natural-image statistics, may favor despite lower pixel accuracy. As all methods optimize the same $\ell_2$ loss, the gap reflects a metric preference: LPIPS can reward texture-like variation over the clean, scale-consistent geometry targeted by low-budget vectorization.

\begin{figure}[!b]
  \centering
  \includegraphics[width=0.6\linewidth]{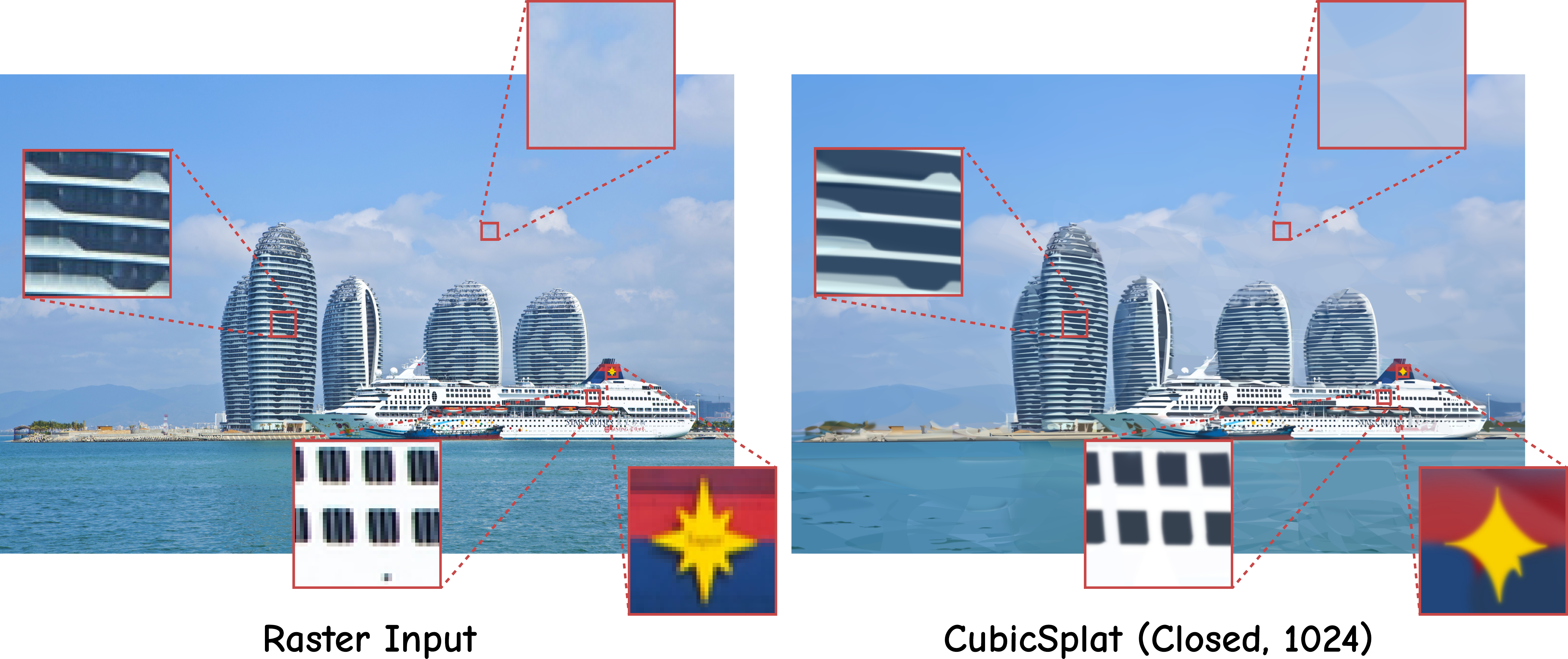}
  \caption{Close-up detail comparison. Raster input (left) and CubicSplat reconstruction with 1024 closed fills (right). Insets highlight preserved sharp architectural edges, coherent layer structure, and clean color boundaries using only parametric vector primitives, yielding stable fine details under magnification.}
  \label{fig:detail}
\end{figure}

\subsection{Ablations and Analysis}
\label{sec:ablation}

\input{tables/ablation}

\Cref{tab:div2k_ablation} isolates four design choices on the uniformly sampled
DIV2K subset under both topology modes and curve budgets $N\in\{256,512,1024\}$.
Taken together, the ablations support our central thesis: \method is an
error-controlled forward relaxation that improves the optimization oracle, and pushing
forward geometric exactness beyond a non-degeneracy threshold yields diminishing returns
in quality while imposing substantial performance costs.

When the curve budget or sampling density increases, additional work can be partially
absorbed by otherwise idle execution resources, producing sublinear scaling in observed
end-to-end time.
Reporting training throughput in tasks-per-minute (TPM) makes this effect explicit and
enables fair comparison across settings where hardware utilization differs.

\textbf{Sampling density $S$ validates forward relaxation.}
Increasing $S$ monotonically reduces geometric surrogate error and, in the limit, recovers an \textbf{exact} B\'{e}zier renderer; yet the learned solutions exhibit a broad equivariance regime once the surrogate is non-degenerate. In open mode, gains saturate beyond $S{=}24$ with negligible throughput change. In closed mode, saturation occurs even earlier because winding-number membership already provides exact interior labeling; $S{=}12$ attains essentially the same quality as $S{=}24$ at substantially higher throughput. These trends confirm that $S$ functions as a relaxation knob rather than an approximation budget, consistent with our error bound (Appendix A.4), which guarantees
decreasing forward bias but does not imply that optimization quality must improve once the oracle is well conditioned.

\textbf{Visibility-driven Gini pruning is essential at large $N$.}
Removing Gini pruning degrades PSNR increasingly with larger curve budgets in both
topology modes.
This aligns with the oracle perspective: as $N$ grows, occlusion depth increases and
more primitives become weakly visible or redundant, amplifying gradient interference
and wasting compute.
By explicitly tracking contribution and visibility and pruning low-importance primitives,
\method maintains an effective representation size and preserves gradient quality, with
the benefit growing as $N$ increases.

\textbf{Regularization reduces expressivity and is unnecessary under oracle stability.}
Adding a convexity-oriented regularizer consistently decreases reconstruction quality
across all budgets and both modes.
This indicates that the optimal vector representation for natural images is often
non-convex and geometrically complex; restricting the hypothesis class harms
expressivity without a compensating stability benefit.
In \method, stability derives from the relaxed oracle and visibility-aware adaptation:
pathological or redundant curves are automatically identified by low visibility and
pruned, making an explicit regularizer both unnecessary and counterproductive.

\textbf{Adaptive subdivision underperforms due to graph discontinuities.}
De Casteljau subdivision reduces forward geometric error but introduces discrete split events that alter the computation graph. The ablation shows that this strategy trails fixed-$S$ sampling in both quality and throughput, consistent with our analysis that subdivision boundaries induce gradient discontinuities and heteroscedastic updates (Appendix B). Fixed-$S$ relaxation preserves a stable graph and yields more reliable gradients, even when its forward approximation error is formally larger.

\subsection{Scaling Behavior}
\label{sec:scaling}

The error-controlled relaxation analyzed in \cref{sec:ablation} has direct consequences for scalability: because a low sampling density $S$ is not an engineering compromise but a mathematically bounded surrogate with negligible quality loss, CubicSplat can allocate its compute budget almost entirely to geometric expressiveness (more primitives, higher resolution) rather than to tighter curve approximation.

\begin{figure}[!b]
  \centering
  \includegraphics[width=0.9\linewidth]{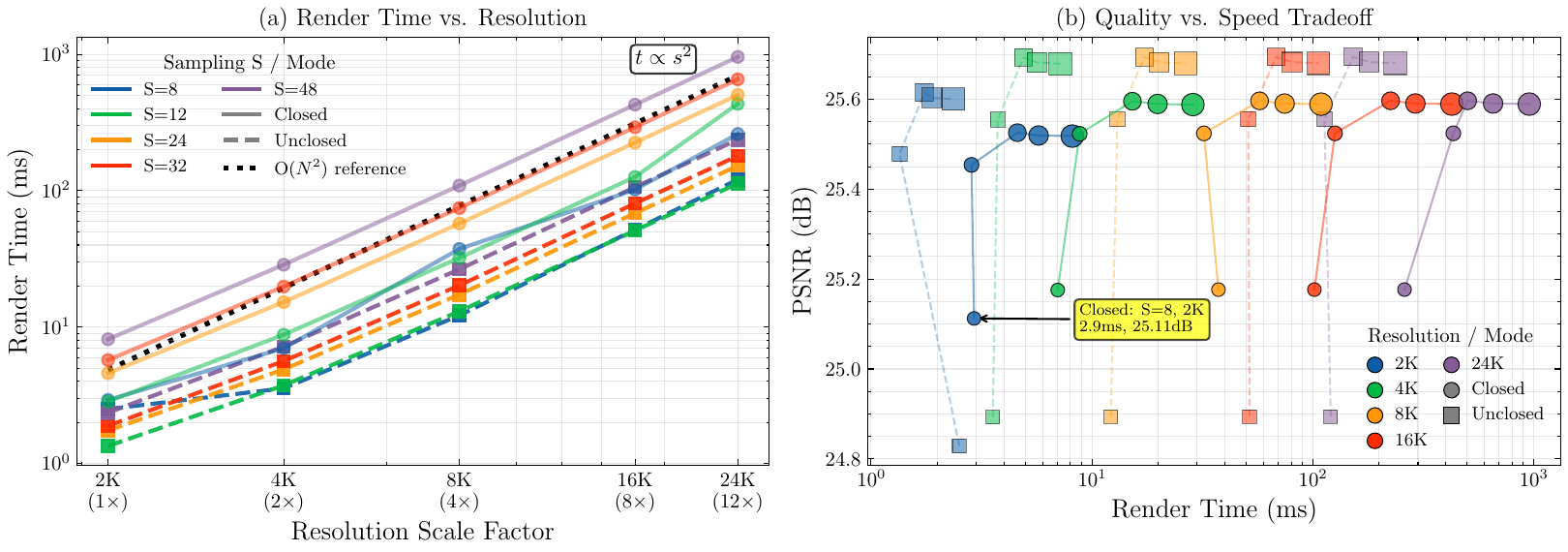}
  \caption{Resolution scalability of CubicSplat. (a)~Render time versus resolution scale at different sampling densities $S$. (b)~PSNR versus render time for all tested (resolution,~$S$) configurations. At native 2K, $S{=}8$ achieves near-optimal quality in 2.9\,ms; increasing $S$ to 48 yields $<$0.03\,dB gain at $4{\times}$ the cost. Original data refers to Appendix D.}
  \label{fig:speed}
\end{figure}

\begin{figure}[!b]
  \centering
  \includegraphics[width=0.95\linewidth]{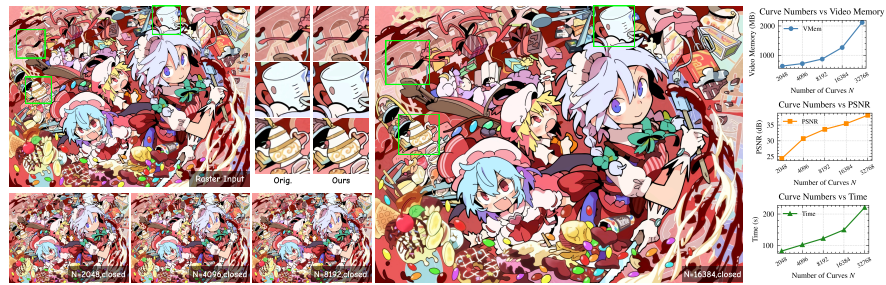}
  \caption{Scaling with closed-fill count $N$ on a detailed illustration from the DanbooRegion dataset~\cite{DanbooRegion2020}. Insets rendered at \textbf{8$\times$ magnification} show \textbf{scale-consistent} fine structures under extreme primitive counts. PSNR improves steadily with sublinear training time and manageable memory overhead.}
  \label{fig:cartoon}
\end{figure}

\textbf{Quality is resolution-invariant up to extreme magnification.}
We train on 5 DIV2K~\cite{Timofte_2017_CVPR_Workshops} images using
1024 curves in both modes and evaluate forward rendering up to
$144{\times}$ more pixels on a single RTX~4090, computing PSNR after
downsampling to original resolution (Appendix D).
Render time scales as $t \propto s^2$ with resolution scale $s$
 in \cref{fig:speed}(a), while PSNR stays within 0.5\,dB of the 2K baseline
across all resolutions, sampling densities, and both modes.
\Cref{fig:speed}(b) confirms that quality is governed by spline geometry
rather than surrogate discretization: increasing $S$ from 12 to 48 in
closed mode changes PSNR by less than 0.03\,dB at ${\sim}4{\times}$ the
cost, consistent with the $\mathcal{O}(S^{-2})$ bound. At native 2K,
$S{=}8$ reaches near-saturated quality at 2.9\,ms (25.11\,dB, closed),
leaving headroom for real-time use. This robustness contrasts with
point-sampled splatting, where discrete sampling introduces
phase-sensitive artifacts under magnification.

\textbf{Optimization remains well conditioned at high primitive counts.}
We vary the number of closed curves $N$ from 2048 to 32768 on a detailed illustration (\cref{fig:cartoon}). PSNR improves consistently from 24.45\,dB to 38.09\,dB, indicating that the gradient oracle does not degrade as scene complexity grows. Training time scales sublinearly (82.5\,s to 220.8\,s) thanks to the tile-parallel architecture and visibility-based pruning, while video memory grows from 632\,MB to 2118\,MB. The combination of low per-primitive cost from the polyline surrogate and effective pruning from compositing-derived visibility statistics allows CubicSplat to operate comfortably at $32\mathrm{K}$ primitives on a single consumer GPU.

%% file: tables/speed.tex
\begin{table}[!t]
\centering
\caption{Single training-step time averaged over a $2040{\times}1344$ image
with 2048 curves on an RTX~4090.}
\vspace{-0.7em}
\resizebox{0.65\textwidth}{!}{
\begin{tabular}{l|ccc|ccc}
\toprule
& \multicolumn{3}{c|}{\textbf{Open curves}} & \multicolumn{3}{c}{\textbf{Closed curves}} \\
& DiffVG & \BezierSplat & Ours & DiffVG & \BezierSplat & Ours \\
\midrule
Forward & 141.3ms & 4.5ms & \textbf{2.70ms} & 85.2ms & 14.1ms & \textbf{10.68ms} \\
Backward & 701.3ms & 4.7ms & \textbf{4.07ms} & 448.3ms & 24.58ms & \textbf{8.96ms} \\
\midrule
Total step & 842.6ms & 9.2ms & \textbf{6.77ms} & 533.5ms & 38.68ms & \textbf{19.64ms} \\
\bottomrule
\end{tabular}}
\label{tab:speed}
\vspace{-0.5em}
\end{table}

%% file: tables/comparison.tex
\begin{table}[tp]
\centering
\caption{Quantitative evaluation on the DIV2K dataset~\cite{Timofte_2017_CVPR_Workshops} (200 images). Best results in \textbf{bold}, second-best \underline{underlined}. Training time is the average per image on an RTX 4090. \method achieves the highest PSNR and SSIM across all budgets and both modes while training in under 2.1\,min per image; in closed mode at $N{=}1024$ it surpasses B\'{e}zier Splatting by over 2\,dB PSNR at roughly $4{\times}$ lower cost.}
\vspace{-0.7em}
\setlength{\tabcolsep}{3pt}
\resizebox{0.99\textwidth}{!}{
\begin{tabular}{l|l|cccc|cccc|cccc}
\toprule
& \multirow{2}{*}{\textbf{Method}}
& \multicolumn{4}{c|}{\textbf{256 curves}}
& \multicolumn{4}{c|}{\textbf{512 curves}}
& \multicolumn{4}{c}{\textbf{1024 curves}} \\
& & SSIM$\uparrow$ & PSNR$\uparrow$ & LPIPS$\downarrow$ & \cellcolor{timegray}Time$\downarrow$
  & SSIM$\uparrow$ & PSNR$\uparrow$ & LPIPS$\downarrow$ & \cellcolor{timegray}Time$\downarrow$
  & SSIM$\uparrow$ & PSNR$\uparrow$ & LPIPS$\downarrow$ & \cellcolor{timegray}Time$\downarrow$ \\
\midrule
\multirow{3}{*}{Open} & DiffVG~\cite{Li:2020:DVG} & 0.552 & 19.83 & \underline{0.563} & \cellcolor{timegray}18.9min & 0.587 & 21.47 & \underline{0.537} & \cellcolor{timegray}22.0min & 0.616 & 22.62 & 0.517 & \cellcolor{timegray}30.6min \\
 & Bézier Splatting~\cite{liu2025bezier} & \underline{0.600} & \underline{22.17} & \textbf{0.540} & \cellcolor{timegray}\underline{3.4min} & \underline{0.646} & \underline{23.79} & \textbf{0.498} & \cellcolor{timegray}\underline{3.3min} & \underline{0.699} & \underline{25.45} & \textbf{0.448} & \cellcolor{timegray}\underline{3.2min} \\
 & \method{} (Ours, $S=24$) & \cellcolor{oursblue}\textbf{0.624} & \cellcolor{oursblue}\textbf{23.07} & \cellcolor{oursblue}0.589 & \cellcolor{timegray}\textbf{1.8min} & \cellcolor{oursblue}\textbf{0.662} & \cellcolor{oursblue}\textbf{24.43} & \cellcolor{oursblue}0.548 & \cellcolor{timegray}\textbf{1.8min} & \cellcolor{oursblue}\textbf{0.708} & \cellcolor{oursblue}\textbf{25.93} & \cellcolor{oursblue}\underline{0.502} & \cellcolor{timegray}\textbf{1.8min} \\
\midrule
\multirow{5}{*}{Closed} & DiffVG~\cite{Li:2020:DVG} & 0.578 & 20.69 & 0.548 & \cellcolor{timegray}16.4min & 0.601 & 21.82 & 0.531 & \cellcolor{timegray}18.5min & 0.631 & 22.95 & 0.509 & \cellcolor{timegray}25.1min \\
 & LIVE~\cite{xu2022live} & 0.576 & 20.09 & \underline{0.543} & \cellcolor{timegray}2.6h & 0.611 & 21.70 & \textbf{0.521} & \cellcolor{timegray}4.2h & 0.648 & 23.11 & \textbf{0.495} & \cellcolor{timegray}5.1h \\
 & LIVSS~\cite{livss} & \underline{0.586} & 17.71 & \textbf{0.542} & \cellcolor{timegray}39.2min & \underline{0.630} & 18.71 & 0.530 & \cellcolor{timegray}54.3min & \underline{0.678} & 19.83 & 0.517 & \cellcolor{timegray}1.4h \\
 & Bézier Splatting~\cite{liu2025bezier} & 0.580 & \underline{20.74} & 0.546 & \cellcolor{timegray}\underline{7.8min} & 0.607 & \underline{22.11} & \underline{0.528} & \cellcolor{timegray}\underline{8.3min} & 0.639 & \underline{23.45} & \underline{0.507} & \cellcolor{timegray}\underline{8.6min} \\
 & \method{} (Ours, $S=12$) & \cellcolor{oursblue}\textbf{0.629} & \cellcolor{oursblue}\textbf{23.00} & \cellcolor{oursblue}0.580 & \cellcolor{timegray}\textbf{1.8min} & \cellcolor{oursblue}\textbf{0.664} & \cellcolor{oursblue}\textbf{24.31} & \cellcolor{oursblue}0.547 & \cellcolor{timegray}\textbf{1.9min} & \cellcolor{oursblue}\textbf{0.705} & \cellcolor{oursblue}\textbf{25.78} & \cellcolor{oursblue}0.508 & \cellcolor{timegray}\textbf{2.1min} \\
\bottomrule
\end{tabular}}
\label{tab:comparison}
\end{table}

\begin{table}[t]
    \centering
    \caption{
        A quantitative evaluation on the Kodak dataset~\cite{kodak1999} with 256 to 1024 curves on RTX 4090. CubicSplat leads in PSNR and SSIM across all settings and achieves the best LPIPS in closed mode.
    }
    \resizebox{0.85\textwidth}{!}{
        \begin{tabular}{l|l|ccc|ccc|ccc}
            \toprule
            \multirow{2}{*} & \multirow{2}{*}{\textbf{Method}} 
             & \multicolumn{3}{c|}{\textbf{256}} 
             & \multicolumn{3}{c|}{\textbf{512}} 
             & \multicolumn{3}{c}{\textbf{1024}} 
              \\
             
             & & $SSIM^\uparrow$ & $PSNR^\uparrow$ & $LPIPS^\downarrow$
               & $SSIM^\uparrow$ & $PSNR^\uparrow$ & $LPIPS^\downarrow$
               & $SSIM^\uparrow$ & $PSNR^\uparrow$ & $LPIPS^\downarrow$ \\
            \midrule

            \multirow{3}{*}{Open} 
            & DiffVG & 0.601 & 23.43 & 0.535  
                    & 0.645 & 24.70 & 0.495  
                    & 0.699 & 26.04 & 0.439   \\
            & Bézier Splatting   & 0.679 & 26.18 & \textbf{0.457}
                    & \textbf{0.743} & 27.90 & \textbf{0.383}
                    & 0.797 & 29.24 & \textbf{0.310}  \\
            & CubicSplat(Ours, S=24)   & \cellcolor{oursblue}\textbf{0.681} & \cellcolor{oursblue}\textbf{26.75} & \cellcolor{oursblue}0.498  
                    & \cellcolor{oursblue}0.740 & \cellcolor{oursblue}\textbf{28.28} & \cellcolor{oursblue}0.432 
                    & \cellcolor{oursblue}\textbf{0.803} & \cellcolor{oursblue}\textbf{29.94} & \cellcolor{oursblue}0.355  \\

            \midrule
            \multirow{3}{*}{Closed} 
            & DiffVG & 0.622 & 24.11 & 0.513
                    & 0.666 & 25.34 & 0.475  
                    & 0.719 & 26.66 & 0.420  \\
            & Bézier Splatting 
                    & 0.621 & 24.19 & 0.519 
                    & 0.664 & 25.61 & 0.485 
                    & 0.708 & 26.91 & 0.448  \\
            & CubicSplat(Ours, S=12)
                    & \cellcolor{oursblue}\textbf{0.680} & \cellcolor{oursblue}\textbf{26.64} & \cellcolor{oursblue}\textbf{0.507}
                    & \cellcolor{oursblue}\textbf{0.730} & \cellcolor{oursblue}\textbf{28.06} & \cellcolor{oursblue}\textbf{0.454}
                    & \cellcolor{oursblue}\textbf{0.786} & \cellcolor{oursblue}\textbf{29.66} & \cellcolor{oursblue}\textbf{0.391}  \\
            \bottomrule
        \end{tabular}
    }
    \label{tab:comparisons_kodak}
\end{table}


%% file: tables/ablation.tex
\begin{table*}[!b]
\centering
\caption{Ablation on the uniformly sampled DIV2K subset~\cite{Timofte_2017_CVPR_Workshops} (200 images). We report SSIM/PSNR/LPIPS and training throughput in Tasks Per Minute of a single GPU(TPM, higher is better), measured with 2 workers per GPU. Best results in \textbf{bold}, second-best \underline{underlined}.}
\vspace{-0.7em}
\setlength{\tabcolsep}{3pt}
\resizebox{0.99\textwidth}{!}{
\begin{tabular}{l|l|cccc|cccc|cccc}
\toprule
& \multirow{2}{*}{\textbf{Setting}}
& \multicolumn{4}{c|}{\textbf{256 curves}}
& \multicolumn{4}{c|}{\textbf{512 curves}}
& \multicolumn{4}{c}{\textbf{1024 curves}} \\
& & SSIM$\uparrow$ & PSNR$\uparrow$ & LPIPS$\downarrow$ & \cellcolor{timegray}TPM$\uparrow$
  & SSIM$\uparrow$ & PSNR$\uparrow$ & LPIPS$\downarrow$ & \cellcolor{timegray}TPM$\uparrow$
  & SSIM$\uparrow$ & PSNR$\uparrow$ & LPIPS$\downarrow$ & \cellcolor{timegray}TPM$\uparrow$ \\
\midrule
\multirow{7}{*}{Open} & $S=24$ & \textbf{0.624} & \textbf{23.07} & \textbf{0.589} & \cellcolor{timegray}0.861 & \textbf{0.662} & \textbf{24.43} & \textbf{0.548} & \cellcolor{timegray}0.883 & \textbf{0.708} & \underline{25.93} & \textbf{0.502} & \cellcolor{timegray}0.873 \\
 & \quad \textbf{w.} Regularization & 0.614 & 22.60 & 0.597 & \cellcolor{timegray}\underline{0.894} & 0.650 & 23.89 & 0.561 & \cellcolor{timegray}\underline{0.938} & 0.693 & 25.28 & 0.522 & \cellcolor{timegray}\underline{0.917} \\
 & \quad \textbf{w.o} Gini-Pruning & 0.611 & 22.38 & 0.596 & \cellcolor{timegray}\textbf{0.971} & 0.645 & 23.56 & 0.554 & \cellcolor{timegray}\textbf{0.982} & 0.686 & 24.83 & 0.508 & \cellcolor{timegray}\textbf{0.952} \\
 & $S=32$ & \textbf{0.624} & \underline{23.06} & \textbf{0.589} & \cellcolor{timegray}0.859 & \textbf{0.662} & \underline{24.42} & \textbf{0.548} & \cellcolor{timegray}0.868 & \textbf{0.708} & \textbf{25.94} & \textbf{0.502} & \cellcolor{timegray}0.830 \\
 & $S=12$ & 0.622 & 22.93 & 0.591 & \cellcolor{timegray}0.814 & 0.658 & 24.27 & 0.550 & \cellcolor{timegray}0.794 & 0.704 & 25.75 & 0.504 & \cellcolor{timegray}0.832 \\
 & $S=8$ & 0.615 & 22.43 & 0.595 & \cellcolor{timegray}0.753 & 0.649 & 23.64 & 0.556 & \cellcolor{timegray}0.775 & 0.692 & 24.97 & 0.511 & \cellcolor{timegray}0.832 \\
 & $\text{de Casteljau}$ & 0.623 & 23.01 & 0.590 & \cellcolor{timegray}0.645 & 0.660 & 24.37 & 0.549 & \cellcolor{timegray}0.631 & 0.706 & 25.87 & 0.503 & \cellcolor{timegray}0.633 \\
\midrule
\multirow{7}{*}{Closed} & $S=24$ & \textbf{0.630} & \textbf{23.01} & \textbf{0.580} & \cellcolor{timegray}0.572 & \textbf{0.664} & \textbf{24.33} & \textbf{0.547} & \cellcolor{timegray}0.481 & \textbf{0.706} & \textbf{25.80} & \textbf{0.508} & \cellcolor{timegray}0.413 \\
 & \quad \textbf{w.} Regularization & 0.620 & 22.63 & 0.588 & \cellcolor{timegray}0.650 & 0.653 & 23.95 & 0.558 & \cellcolor{timegray}0.578 & 0.694 & 25.38 & 0.521 & \cellcolor{timegray}0.503 \\
 & \quad \textbf{w.o} Gini-Pruning & 0.616 & 22.21 & 0.589 & \cellcolor{timegray}0.546 & 0.643 & 23.22 & 0.561 & \cellcolor{timegray}0.456 & 0.672 & 24.21 & 0.531 & \cellcolor{timegray}0.390 \\
 & $S=32$ & \textbf{0.630} & \textbf{23.01} & \textbf{0.580} & \cellcolor{timegray}0.461 & \textbf{0.664} & \textbf{24.33} & \textbf{0.547} & \cellcolor{timegray}0.387 & \textbf{0.706} & \textbf{25.80} & \textbf{0.508} & \cellcolor{timegray}0.317 \\
 & $S=12$ & 0.629 & 22.99 & \textbf{0.580} & \cellcolor{timegray}\underline{0.879} & \textbf{0.664} & 24.32 & \textbf{0.547} & \cellcolor{timegray}\underline{0.786} & 0.705 & 25.78 & 0.509 & \cellcolor{timegray}\underline{0.654} \\
 & $S=8$ & 0.628 & 22.95 & 0.581 & \cellcolor{timegray}\textbf{0.975} & 0.662 & 24.25 & 0.548 & \cellcolor{timegray}\textbf{0.914} & 0.704 & 25.70 & 0.510 & \cellcolor{timegray}\textbf{0.826} \\
 & $\text{de Casteljau}$ & 0.629 & 22.97 & 0.581 & \cellcolor{timegray}0.517 & 0.663 & 24.29 & 0.548 & \cellcolor{timegray}0.531 & 0.704 & 25.73 & 0.510 & \cellcolor{timegray}0.481 \\
\bottomrule
\end{tabular}}
\label{tab:div2k_ablation}
\end{table*}

%% file: sec/5_conclusion.tex
\section{Conclusion}

We introduce CubicSplat, a scalable differentiable vector rasterizer that addresses the gradient seesaw through error-bounded forward relaxation. By replacing B\'{e}zier closest-point root solving with closed-form polyline distance queries, CubicSplat yields a stable, static computation graph while maintaining a tight $\mathcal{O}(S^{-2})$ geometric error bound. Experiments on large-scale benchmarks confirmed that this controlled relaxation does not sacrifice reconstruction quality: the better-conditioned gradient oracle consistently produces superior optimization outcomes at a fraction of the computational cost.

These efficiency characteristics suggest two natural extensions. The tile-parallel architecture and lightweight per-primitive cost bring \textbf{real-time differentiable vector rendering within reach}, enabling interactive design tools that optimize on the fly at display rates\cite{das2020beziersketch}. The low per-frame cost also \textbf{opens the door to video vectorization}, where temporal coherence demands both fast rendering and stable gradients under continuous motion. More broadly, the gradient seesaw perspective may offer a useful lens for other differentiable rendering settings where forward exactness is pursued at the expense of optimization quality. Combining CubicSplat with perceptual losses, semantic constraints, and learned priors~\cite{jain2023vectorfusion,xing2024svgdreamer,li20254dlangsplat4dlanguage,Du:2023:IVE,Chen_2024_CVPR,hirschorn2024optimize,jain2023vectorfusion,xing2023diffsketcher,zhang2024text,Cao_2023_CVPR} is a promising direction toward controllable, real-time vector-content creation.

%% file: sec/6_appendix.tex
\section{Formal Analysis of Differentiability and Gradient Structure}%
\label{app:theory}

For convenience we restate the two definitions from the main paper that
are used throughout this supplement.
\paragraph{Soft coverage.}
Each primitive $p$ contributes a per-pixel opacity
\begin{equation}
\label{eq:supp_pipeline}
  \alpha_p(x)
  \;=\;
  \phi\!\left(
    \frac{\delta_p(x,\Pi_S(\mathcal{B}_p))}{\tau_p}
  \right),
\end{equation}
where $\phi$ is a coverage kernel, $\tau_p>0$ is a softness parameter,
and $\delta_p$ is a \emph{signed residual} defined as
\begin{equation}
\label{eq:supp_signed_residual}
\delta_p(x,\, \Pi_S(\mathcal{B}_p)) \;=\;
\begin{cases}
  r_p \;-\; d(x,\,\mathcal{P}_p),
    & \text{open stroke},\\[4pt]
  s_p(x)\;\cdot\; d(x,\,\mathcal{P}_p),
    & \text{closed fill},
\end{cases}
\end{equation}
with $d(x,\mathcal{P}_p)$ the unsigned distance from pixel $x$ to the
path $\mathcal{P}_p$, $r_p$ the half-stroke width, and $s_p(x)\in\{+1,-1\}$
the inside/outside sign determined by the winding rule.
The final raster image is obtained by compositing
$R_S(\theta)=\mathcal{C}\!\bigl(\{\alpha_p, c_p\}_{p=1}^{N}\bigr)$
over all $N$ primitives in depth order.

\begin{remark}[One-sided anti-aliasing in fill mode]
\label{rem:onesided_aa}
For closed fills, the signed residual $\delta_p(x)=s_p(x)\cdot
d(x,\mathcal{P}_p)$ is positive inside and negative outside
(\cref{eq:supp_signed_residual}). Because $\phi$ saturates to $1$ for large
positive arguments, interior pixels far from the boundary receive
$\alpha_p\approx 1$ irrespective of $d$; the smooth $1\to 0$ transition
is carried entirely by pixels near or outside the boundary, where
$\delta_p\approx 0$ or $\delta_p<0$. This one-sided convention avoids
double-counting coverage at shared edges and is consistent with standard
winding-number rasterizers. All experiments use the nonzero winding rule.
\end{remark}

\subsection{Preliminaries}

A single B\'ezier segment of degree $M$ is defined as:
\begin{equation}
\mathcal{B}_i(t)
= \sum_{j=0}^{M} B_j^M(t)\,P^{(i)}_j,
\qquad t\in[0,1],\; i\in\{1,\dots,N\},
\label{eq:bezier}
\end{equation}
with control points $P^{(i)}_j\in\mathbb{R}^2$ and Bernstein basis:
\begin{equation}
B_j^M(t)=\binom{M}{j}(1-t)^{M-j}\cdot t^j.
\label{eq:bezier1}
\end{equation}
Each primitive additionally carries appearance parameters: RGB color
$c_p\in\mathbb{R}^3$, opacity $o_p\in[0,1]$, and coverage kernel
parameter $\tau_p$. Strokes also have a learnable half-width $w_i>0$.
We collect all learnable parameters in $\theta\in\Theta$.

\begin{definition}[Lipschitz continuity]
A function $f:\mathbb{R}^m\to\mathbb{R}^n$ is $L$-Lipschitz if
$\|f(u)-f(v)\|_2\le L\|u-v\|_2$ for all $u,v\in\mathbb{R}^m$.
It is \emph{locally Lipschitz} if every point has a neighborhood on
which this holds for some finite $L$.
\end{definition}

\begin{theorem}[Rademacher~\cite{rademacher1919}]
Every locally Lipschitz function $f:\mathbb{R}^m\to\mathbb{R}^n$ is
differentiable almost everywhere with respect to Lebesgue measure.
\end{theorem}

\begin{theorem}[Danskin--envelope theorem for a minimum;
{\cite{danskin1967,bertsekas1971}}]
\label{thm:danskin}
Let $g:\mathcal{T}\times\Theta\to\mathbb{R}$ be continuously
differentiable in $\theta$ for each $t\in\mathcal{T}$, and define
$\phi(\theta)=\min_{t\in\mathcal{T}}g(t,\theta)$.
\begin{enumerate}
  \item \emph{(Unique minimizer.)} If the minimizer
        $t^\star(\theta)\in\mathcal{T}$ is unique and lies in the
        interior of $\mathcal{T}$, then $\phi$ is differentiable at
        $\theta$ and
        $\nabla_\theta\phi(\theta)=\nabla_\theta g(t^\star(\theta),\theta)$.
  \item \emph{(Multiple minimizers.)} If
        $\mathcal{T}^\star(\theta)=\arg\min_{t\in\mathcal{T}}g(t,\theta)$
        contains more than one element, then $\phi$ is generally
        nonsmooth; its Clarke subdifferential~\cite{clarke1990} satisfies
        $\partial^C\phi(\theta)\subseteq
        \operatorname{conv}\{\nabla_\theta g(t,\theta):
        t\in\mathcal{T}^\star(\theta)\}$.
\end{enumerate}
\end{theorem}
\subsection{Ill-Conditioning of the DiffVG Analytical Prefilter}

\subsubsection{Closest-point distance as a value function.}

Fix a pixel center $x\in\Omega$ and let $b(t;\theta):[0,1]\to\mathbb{R}^2$
be a cubic B\'ezier curve with control points $\theta$. Define the squared
distance
\begin{equation}
  g(t,\theta)
  \;=\;
  \|b(t;\theta)-x\|_2^2,
  \qquad
  \phi(\theta)
  \;=\;
  \min_{t\in[0,1]} g(t,\theta).
  \label{eq:diffvg_value}
\end{equation}
When the minimizer $t^\star(\theta)$ is unique and interior to $[0,1]$,
\cref{thm:danskin}(1) gives
\begin{equation}
  \nabla_\theta\phi(\theta)
  \;=\;
  \nabla_\theta g\!\bigl(t^\star(\theta),\theta\bigr)
  \;=\;
  2\bigl(b(t^\star;\theta)-x\bigr)^\top\!\nabla_\theta b(t^\star;\theta).
  \label{eq:envelope_grad}
\end{equation}
The signed distance
$d_{\mathrm{signed}}(\theta)=\pm\sqrt{\phi(\theta)}$, where the sign is
determined by the winding number as in DiffVG, is then differentiable at
all $\theta$ with $\phi(\theta)>0$ by the chain rule.

\subsubsection{Intrinsic nonsmoothness under multiple closest points.}

When the minimizer set $\mathcal{T}^\star(\theta)$ contains more than one
point, $\phi$ is generally not differentiable. By \cref{thm:danskin}(2),
the Clarke subdifferential satisfies
\begin{equation}
  \partial^C\phi(\theta)
  \;\subseteq\;
  \operatorname{conv}\!\Bigl\{
    2\bigl(b(t;\theta)-x\bigr)^\top\!\nabla_\theta b(t;\theta)
    \;:\;
    t\in\mathcal{T}^\star(\theta)
  \Bigr\}.
  \label{eq:clarke_value}
\end{equation}
Non-unique closest points arise when a pixel is equidistant to two or
more curve locations, a situation that occurs near medial axes, cusps,
self-intersections, and near-parallel boundaries. Such configurations are
unavoidable in practice and explain why analytical prefiltering exhibits
gradient conflation and optimization instability without additional
regularization.

\subsubsection{Gradient error under approximate root solving.}
\label{app:sub:diffvg_rooterror}

Finding $t^\star(\theta)$ requires solving
$\partial_t g(t,\theta)=0$, a degree-5 polynomial in
$t$~\cite{Li:2020:DVG}. In practice, only an approximate stationary
point $\hat{t}(\theta)$ is obtained
numerically~\cite{Li:2020:DVG}. The surrogate value is
$\hat\phi(\theta)=g(\hat{t}(\theta),\theta)$, whose total derivative is
\begin{equation}
  \nabla_\theta\hat\phi(\theta)
  \;=\;
  \underbrace{\nabla_\theta g(\hat{t},\theta)}_{\approx\,\nabla_\theta\phi(\theta)}
  \;+\;
  \underbrace{\partial_t g(\hat{t},\theta)\,\nabla_\theta\hat{t}(\theta)}_{\text{residual}}.
  \label{eq:approx_total}
\end{equation}

Many implementations treat $\hat t$ as a non-differentiable argmin and apply a stop-gradient through the numerical solver, retaining only the first term $\nabla_\theta g(\hat t,\theta)$; the resulting error then depends on the stationarity defect $|\partial_t g(\hat t,\theta)|$.

At a true stationary point $\partial_t g(t^\star,\theta)=0$, the residual vanishes and \cref{eq:envelope_grad} is recovered. In floating-point arithmetic, however, $\partial_t g(\hat{t},\theta)\ne 0$ in general, introducing a gradient error bounded by
\begin{equation}
  \bigl\|\nabla_\theta\hat\phi(\theta)-\nabla_\theta\phi(\theta)\bigr\|
  \;\lesssim\;
  |\partial_t g(\hat{t},\theta)|\;\|\nabla_\theta\hat{t}(\theta)\|
  \;+\;
  \text{(quadrature error in }\nabla_\theta g\text{)}.
  \label{eq:grad_error_bound}
\end{equation}
When $\hat{t}$ is obtained by differentiating through the implicit root
condition $p(t,\theta)\equiv\partial_t g(t,\theta)=0$ via the implicit
function theorem,
\begin{equation}
  \nabla_\theta\hat{t}(\theta)
  \;=\;
  -\frac{\nabla_\theta p(\hat{t},\theta)}{\partial_t p(\hat{t},\theta)}.
  \label{eq:ift_root}
\end{equation}
This expression becomes ill-conditioned whenever
$\partial_t p(\hat{t},\theta)\approx 0$, i.e.\ when roots are nearly multiple or the distance landscape is nearly flat (e.g.\ a pixel roughly equidistant to two curve segments). In such configurations the residual term in \cref{eq:approx_total} dominates, producing large gradient noise that directly drives optimization instability. This provides a formal account of why DiffVG requires either additional regularization near geometric degeneracies or a fallback to its
Monte Carlo mode at substantially higher computational cost.
\subsection{Discretization Artifacts in Point-Sampled Curve Splatting}
\label{app:sub:bgs}

This subsection formalizes a key limitation of point-sampled splatting
approaches such as B\'ezier Splatting~\cite{liu2025bezier}, which places
Gaussian primitives at discrete samples along a B\'ezier curve. The core
issue is not merely approximation error in a norm, but the introduction
of phase-sensitive aliasing that breaks scale equivariance and makes both
forward appearance and backward gradients sensitive to sampling
schedules.

\begin{figure}[h]
    \centering
    \includegraphics[width=0.8\linewidth]{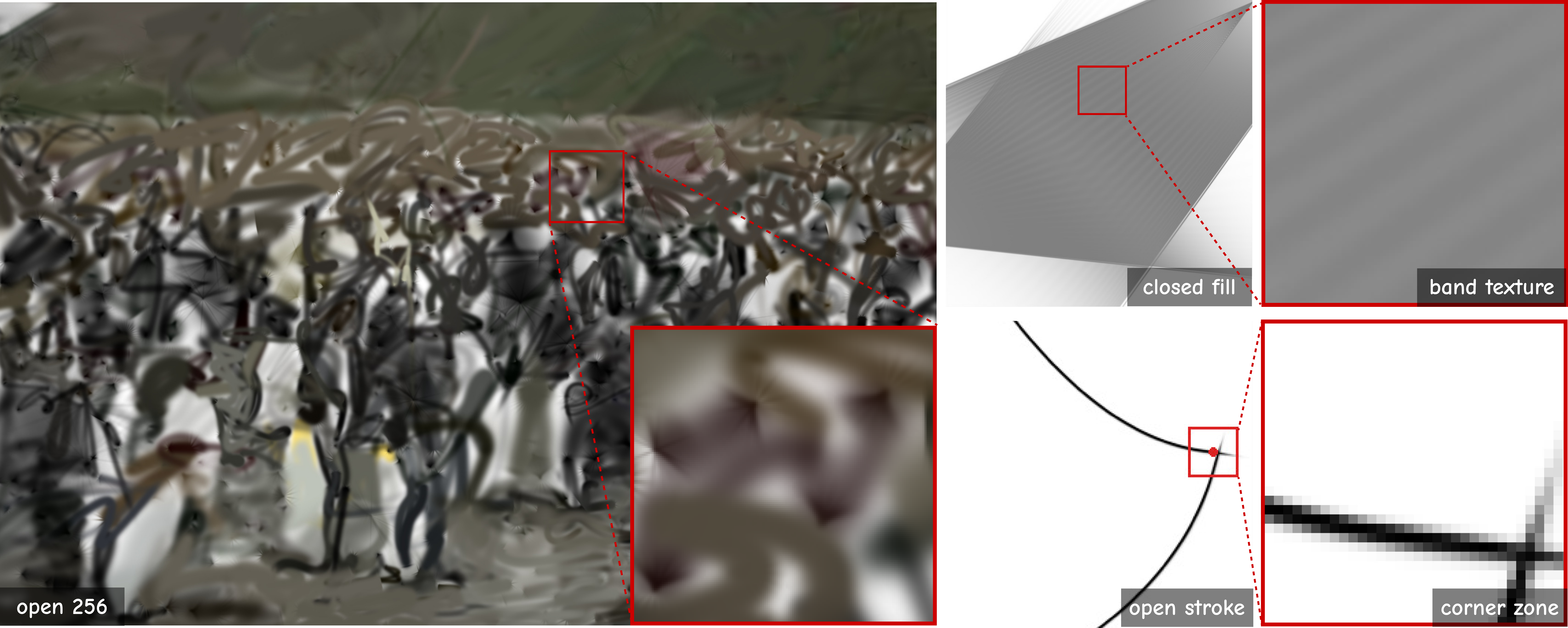}
    \caption{\textbf{Single-primitive artifacts of B\'ezier Splatting.}
We isolate failure modes that persist even without multi-primitive
occlusion or scheduling effects.
\textit{Left:} A typical open-stroke reconstruction (256 curves) exhibits
texture-like blotches arising from sample-based image formation rather
than geometric structure.
\textit{Right (top):} A single closed fill reveals coherent diagonal
banding inside the primitive (zoomed inset), showing that the discretized
Gaussian sum induces structured micro-texture even in the simplest
setting.
\textit{Right (bottom):} A single open stroke shows a pronounced
spike near the endpoint curvature transition (zoomed inset), where local
sampling density and compositing amplify spurious intensity.
These artifacts are intrinsic to discrete splatting and vary with
rescaling or sampling schedules.}
    \label{fig:bgs-onion}
\end{figure}

\subsubsection{Continuous Line Integral vs.\ Discrete Splat Sum.}

Consider a single curve locus $b(t;\theta)$ and an isotropic smooth
kernel $\kappa:\mathbb{R}^2\to\mathbb{R}_+$ (e.g.\ a Gaussian). A
natural continuous coverage field is the line integral
\begin{equation}
  \alpha_{\mathrm{cont}}(x;\theta)
  \;\coloneqq\;
  \int_{0}^{1} \kappa\!\bigl(x-b(t;\theta)\bigr)\,dt.
  \label{eq:cont_integral}
\end{equation}
Point-sampled splatting replaces this with a Riemann sum at step
$h=1/S$:
\begin{equation}
  \alpha_{h}(x;\theta)
  \;\coloneqq\;
  h\sum_{n=0}^{S-1} \kappa\!\bigl(x-b(nh;\theta)\bigr).
  \label{eq:splat_sum}
\end{equation}
While $\alpha_h \to \alpha_{\mathrm{cont}}$ as $h\to 0$, the
discretization introduces a \emph{sampling comb} that yields structured,
phase-dependent high-frequency artifacts at any finite $h$.

\subsubsection{Comb Spectrum and Beat Textures.}

To expose the mechanism, consider a local arc-length parameterization $s$
near a point on the curve, so that $b(t)$ is locally approximated by
$b(s)=b_0+s\,u$ with unit tangent $u$. The continuous model integrates
$\kappa$ over $s$, whereas the discrete model samples it at
$s_n = n\Delta$ with spacing $\Delta$ proportional to $h$ (up to local
speed $\|b'(t)\|$). In this 1D surrogate the discrete sampling multiplies
the continuous integrand by a Dirac comb, whose Fourier transform is
again a comb. The spectrum of $\alpha_h$ therefore contains aliased
replicas whose amplitudes depend on $\Delta$ and whose \emph{phase}
depends on the offset of sample locations relative to the pixel grid.

\begin{proposition}[Phase-sensitive aliasing]
\label{prop:phase_alias}
Assume $\kappa$ is smooth and not strictly band-limited (e.g.\ a
Gaussian), and consider a region where the curve is locally
well-approximated by a straight segment. Then
$\alpha_{h}(x;\theta)-\alpha_{\mathrm{cont}}(x;\theta)$ contains
oscillatory components whose spatial frequency is set by the sampling
spacing $\Delta$ and whose amplitude depends on the kernel bandwidth.
The sign and magnitude of these components further depend on the
sampling phase, i.e.\ the sub-pixel offset of the sample lattice
relative to the image grid.
\end{proposition}

\Cref{prop:phase_alias} explains the empirically observed latent
textures: even when the deviation is low-amplitude, it is structured
(banding, beating) and varies non-smoothly with $\Delta$ or sample phase.
This is qualitatively different from the geometric
$\mathcal{O}(S^{-2})$ Hausdorff error of our polyline surrogate
(\cref{sec:errbound}), which does not introduce a sampling comb into the
image formation model.

\subsubsection{Failure of Scale Equivariance}

A particularly damaging consequence is the violation of scale
equivariance. Let $T_s$ denote image scaling by factor $s>0$ (mapping
$x\mapsto sx$) and $S_s$ the corresponding scaling of curve geometry. In
an ideal continuous model, rendering commutes with scaling up to trivial
resampling:
\begin{equation}
  \alpha_{\mathrm{cont}}(sx;\, S_s\theta)
  \;\approx\;
  \alpha_{\mathrm{cont}}(x;\theta).
  \label{eq:scale_ideal}
\end{equation}
The discrete splat model does \emph{not} commute for fixed $h$:
\begin{equation}
  \alpha_h(sx;\, S_s\theta)
  \;\not\approx\;
  \alpha_h(x;\theta),
  \label{eq:scale_break}
\end{equation}
because scaling changes the effective spacing $\Delta$ in image space
while the comb structure remains tied to the discrete sample set. Any
heuristic densification scheme that adjusts $S$ across resolutions or
iterations therefore induces resolution- and schedule-dependent
artifacts, consistent with the micro-textures observed under
magnification.

\begin{corollary}[Sampling-schedule dependence]
\label{cor:schedule}
For point-sampled splatting, both the forward image and the induced loss
surface change under (i)~resolution changes, (ii)~changes in $S$, or
(iii)~redistribution of samples along the curve, even when the
underlying continuous geometry is unchanged. Optimization thus becomes
sensitive to discretization schedules.
\end{corollary}

\subsubsection{Gradient Implications}

The same comb mechanism affects gradients. Differentiating
\cref{eq:splat_sum} with respect to a control point in $\theta$ gives
\begin{equation}
  \nabla_\theta \alpha_h(x;\theta)
  \;=\;
  -h\sum_{n=0}^{S-1}
  \Bigl[\nabla \kappa\!\bigl(x-b(nh;\theta)\bigr)\Bigr]\,
  \nabla_\theta b(nh;\theta).
  \label{eq:splat_grad}
\end{equation}
Because the sample locations $\{b(nh;\theta)\}$ move with $\theta$, the
gradient aggregates high-frequency contributions of $\nabla\kappa$
evaluated at a discrete set whose phase relative to pixel centers shifts
during optimization. Even when the curve geometry evolves smoothly,
$\nabla_\theta \alpha_h$ can therefore exhibit oscillatory,
schedule-dependent behavior at any finite $h$, injecting gradient
variance unrelated to the true continuous objective. This provides a
principled explanation for why such methods require careful tuning of
sampling density and can fail to preserve fine structure under extreme
scaling.

\subsubsection{Perceptual Metric Bias Toward Sampling Artifacts}
\label{app:sub:lpips_bias}
The structured micro-textures introduced by point-sampled splatting
(\cref{prop:phase_alias}) not only degrade geometric fidelity but can
also artificially inflate perceptual quality scores. 

We conduct a controlled ablation on the DIV2K 1024-open
benchmark: all optimized CubicSplat primitives are frozen, and random,
non-target-derived noise, Fourier-shaped with a spectral envelope
$|\omega|^{-1/2}$ at amplitude $0.005$, is injected into the final
rasterized image. This perturbation introduces texture-like
high-frequency content that is entirely uncorrelated with the target
signal. Under perturbation, PSNR and SSIM degrade slightly
($25.93/0.7078 \!\rightarrow\! 25.86/0.6974$), consistent with the
added signal energy; however, LPIPS improves
($0.5016 \!\rightarrow\! 0.4636$), closing much of the gap toward
B\'ezier Splatting ($0.448$). The residual difference is expected:
B\'ezier Splatting produces high-frequency residuals that are spatially
correlated with the underlying geometry through its sample-based
rasterization, whereas our synthetic perturbation is only
spectrum-matched and lacks such geometric correlation.

\subsubsection{Contrast with CubicSplat.}
Our construction avoids a sampling comb in the image formation model: the
polyline surrogate controls geometric approximation via a deterministic
Hausdorff bound (\cref{eq:hausdorff_bound}), and the induced coverage
error is uniformly bounded by \cref{eq:coverage_error}. Empirically,
this manifests as near-invariance of PSNR across a wide range of $S$ and
resolutions (Sec. 4.5), confirming that reconstruction quality
is governed by continuous spline geometry rather than discretization
fineness.

\subsection{Error Bounds of \method}
\label{sec:errbound}

We use the polyline surrogate, coverage definitions, and winding-number
membership from Sec. 3.3. The results below formalize
differentiability and approximation error properties of that
construction.

\begin{lemma}[Distance is 1-Lipschitz in $x$~\cite{rockafellarwets1998}]
\label{lem:1lip}
For any nonempty closed set $S\subset\mathbb{R}^2$, the distance
function $\operatorname{dist}_S(x)=\inf_{y\in S}\|x-y\|_2$ satisfies
$|\operatorname{dist}_S(x)-\operatorname{dist}_S(x')|\le\|x-x'\|_2$
for all $x,x'\in\mathbb{R}^2$.
\end{lemma}

\begin{proof}
For any $y\in S$, the triangle inequality gives
$\|x'-y\|_2\le\|x'-x\|_2+\|x-y\|_2$. Taking infima over $y$ yields
$\operatorname{dist}_S(x')\le\|x'-x\|_2+\operatorname{dist}_S(x)$.
Exchanging $x$ and $x'$ gives the reverse bound.
\end{proof}

\begin{proposition}[Local Lipschitzness and a.e.\ differentiability in
parameters]
\label{prop:ae_diff}
Assume the vertices of $\mathcal{P}_S(\theta)$ are locally Lipschitz in
$\theta$ (which holds because each vertex $b(k/S;\theta)$ is a
polynomial, hence smooth, function of the B\'ezier control points). Then
for any fixed $x\in\Omega$:
\begin{enumerate}
  \item The unsigned surrogate distance $\tilde{d}(x;\theta)$ is locally
        Lipschitz in $\theta$, hence differentiable a.e.\ in $\theta$ by
        Rademacher's theorem.
  \item The winding-number indicator
        $\mathbf{1}_{\mathrm{inside}}(x;\theta)$ is piecewise constant
        in $\theta$: it changes only when a polyline edge crosses the
        horizontal ray from $x$, a codimension-1 event in parameter
        space.
  \item The surrogate coverage $\tilde\alpha(x;\theta)$ (in both stroke
        and fill variants) is therefore locally Lipschitz on the
        complement of this codimension-1 transition set, and
        differentiable a.e.\ in $\theta$.
\end{enumerate}
\end{proposition}

\begin{proof}
\emph{Part~(1).}
The unsigned distance to $\mathcal{P}_S(\theta)$ is the pointwise
minimum of $S$ segment-distance functions. For the $k$-th segment with
endpoints $v_k(\theta),v_{k+1}(\theta)$, the point-to-segment distance
is a composition of locally Lipschitz operations (projection, clamping,
Euclidean norm), hence locally Lipschitz in $\theta$. A pointwise
minimum of finitely many locally Lipschitz functions is itself locally
Lipschitz, so $\tilde d(x;\theta)$ inherits this property.

On any neighborhood where the minimizing segment
$k^\star(\theta)=\arg\min_k d_k(x;\theta)$ is unique and locally
constant, $\tilde d(x;\theta)=d_{k^\star}(x;\theta)$ is classically
differentiable with
$\nabla_\theta\tilde d=\nabla_\theta d_{k^\star}$. Nondifferentiability
occurs only on the transition set where two or more segments attain the
same minimum or where the closest point on a segment switches between an
interior projection and an endpoint. Under genericity assumptions on
$\theta$, these events form sets of Lebesgue measure zero. At such tie
points the Clarke subdifferential equals the convex hull of the active
segment gradients, so any tie-breaking choice selects an admissible
subgradient. The practical effect is confined to small, localized
gradient variance rather than systematic bias. Rademacher's theorem then
yields a.e.\ differentiability. Empirically, training remains well
conditioned across seeds, up to 32K primitives and high resolutions.

\emph{Part~(2).}
The winding number $\mathrm{winding}(x;\theta)\in\mathbb{Z}$ is a sum of signed crossing indicators $w_k\in\{-1,0,+1\}$, each determined by whether the $k$-th segment crosses the horizontal ray from $x$. Each
crossing event defines a codimension-1 surface
$\{\theta: y_k(\theta)=y_x \text{ or } y_{k+1}(\theta)=y_x\}$ in parameter space. Away from these surfaces $w_k$ is constant, so $s_p(x;\theta) = \mathrm{sign}[\mathrm{winding}(x;\theta)]$ is piecewise constant with transitions on a measure-zero set.

\emph{Part~(3).}
On the open complement of the codimension-1 transition set,
$s_p(x)$ is constant. The signed residual
$\delta_p = s_p \cdot \tilde{d}$ is then a locally Lipschitz function
of $\theta$ by Part~(1), and
$\tilde\alpha^{\mathrm{fill}} = \phi(\delta_p/\tau_p)$ inherits this
property from the Lipschitz continuity of $\phi$. The stroke variant
depends only on $r_p - \tilde{d}$ and is locally Lipschitz everywhere.
Rademacher's theorem yields a.e.\ differentiability in both cases.
\end{proof}

\subsubsection{Clarke Structure of Nonsmoothness.}

The nonsmooth points of $\tilde{d}(x;\theta)$ arise from two sources:
(i)~switching of the closest polyline segment as $\theta$ varies, and
(ii)~switching between the interior projection and an endpoint as the
closest point on a given segment. Under genericity assumptions on
$\theta$, both events form measure-zero sets in parameter space. At such
points the Clarke subdifferential equals the convex hull of the active
segment-distance gradients, analogously to
\cref{eq:clarke_value}~\cite{clarke1990}.

\subsubsection{Hausdorff Approximation Error and Its Propagation.}
\label{app:sub:hausdorff}

The deviation between a cubic B\'ezier curve and its $S$-segment
uniform-sample polyline is bounded by the standard chord-deviation
estimate for piecewise-linear approximation of $C^2$
curves~\cite{farin2002curves}. Each segment subtends a parameter interval
of length $h=1/S$; by a second-order Taylor expansion, the maximum
deviation on that segment is at most
$h^2\max_t\|b''(t;\theta)\|_2/8$. Taking the maximum over all segments
yields the Hausdorff bound:
\begin{equation}
  \operatorname{H}\!\bigl(\mathcal{B}(\theta),\mathcal{P}_S(\theta)\bigr)
  \;\le\;
  \varepsilon_S(\theta),
  \qquad
  \varepsilon_S(\theta)
  \;=\;
  \frac{M(\theta)}{8S^2},
  \qquad
  M(\theta)\;\coloneqq\;\max_{t\in[0,1]}\|b''(t;\theta)\|_2,
  \label{eq:hausdorff_bound}
\end{equation}
where $\operatorname{H}$ denotes the Hausdorff distance. Let
$d_{\mathcal{B}}(x;\theta)\coloneqq\operatorname{dist}(x,\mathcal{B}(\theta))$
denote the true unsigned distance to the curve locus. The bound decays
as $\mathcal{O}(S^{-2})$ for fixed curvature.

By the 1-Lipschitz property of distance (\cref{lem:1lip}),
\begin{equation}
  \bigl|d_{\mathcal{B}}(x;\theta)-\tilde{d}(x;\theta)\bigr|
  \;\le\;
  \operatorname{H}\!\bigl(\mathcal{B}(\theta),\mathcal{P}_S(\theta)\bigr)
  \;\le\;
  \varepsilon_S(\theta)
  \quad\text{for all }x.
  \label{eq:dist_error}
\end{equation}
Since $K$ is $L_K$-Lipschitz, the induced coverage error is uniformly
bounded:
\begin{equation}
  |\alpha(x;\theta)-\tilde\alpha(x;\theta)|
  \;\le\;
  L_K\,\varepsilon_S(\theta).
  \label{eq:coverage_error}
\end{equation}

In fill mode, \cref{eq:coverage_error} holds on pixels where the
polyline and true curve agree on the inside/outside label. A sufficient
condition is
$\operatorname{dist}(x,\partial\mathcal{B}(\theta))>\varepsilon_S(\theta)$,
under which the winding number is stable under a Hausdorff perturbation
of size $\varepsilon_S$. Where the labels differ, an additional term of
magnitude at most $1$ appears via
$|\mathbf{1}_{\mathrm{inside}}-\tilde{\mathbf{1}}_{\mathrm{inside}}|$.

\subsubsection{Inexact-Oracle Interpretation.}
\label{app:sub:scale_adapt}

Let $F(\theta)$ and $\tilde{F}_S(\theta)$ denote the objectives computed
with reference coverage $\alpha(x;\theta)$ and surrogate coverage
$\tilde\alpha(x;\theta)$, respectively. From \cref{eq:coverage_error}
and the Lipschitz continuity of $\ell$,
\begin{equation}
  |F(\theta) - \tilde{F}_S(\theta)|
  \;\le\;
  L_\ell\,|\Omega|\,L_K\,\varepsilon_S(\theta)
  \;\eqqcolon\;
  \delta_S(\theta),
  \label{eq:func_approx}
\end{equation}
where $L_\ell$ is the Lipschitz constant of $\ell$. For fixed $S$,
$\nabla\tilde{F}_S(\theta)$ is the gradient of the surrogate and is
a.e.\ well-defined by \cref{prop:ae_diff}. Gradient descent on
$\tilde{F}_S$ thus constitutes a standard first-order method on a smooth
surrogate that is uniformly $\delta_S$-close to $F$.

When $S$ adapts to $\theta$ through discrete decisions (e.g.\
curvature-based refinement), applying a stop-gradient through the
discrete choice is equivalent to selecting, at each iterate, a surrogate
$\tilde{F}_{S(\theta)}$ from the parametric family and taking a valid
gradient step on that surrogate. This viewpoint connects to the
framework of inexact first-order
oracles~\cite{devolder2014}, provided one controls not only the function
approximation error $|F-\tilde F_S|$ but also the gradient mismatch
between $\nabla F$ and $\nabla\tilde F_S$ along the trajectory. Without
such a gradient-approximation bound, the strongest guarantee is that
optimization proceeds on a surrogate uniformly close to $F$ in function
value.
\subsection{Coverage Kernels}
\label{sec:kernels}

The kernel $\phi$ in \cref{eq:supp_pipeline} maps the normalized signed
residual $\delta_p/\tau_p$ (\cref{eq:supp_signed_residual}) to soft coverage.
We instantiate $\phi$ differently for the two primitive topologies.

For open strokes the signed residual is $\delta_p = r_p - d$, and we use
a one-sided Gaussian:
\begin{equation}
\phi_{\text{open}}\!\left(\frac{\delta_p}{\tau_p}\right)
\;=\; \exp\!\left(-\frac{\max(-\delta_p,\,0)^{2}}{2\tau_p^{2}}\right).
\label{eq:kernel_open}
\end{equation}
When $\delta_p \ge 0$ (inside the stroke) coverage equals one; when
$\delta_p < 0$ (outside) it falls off as a Gaussian in $|\delta_p|$.
The clamping at $\delta_p = 0$ ensures $C^{1}$ continuity with
$\phi'(0) = 0$, so gradients vanish smoothly at the stroke boundary
rather than exhibiting a discontinuous jump.

For closed fills the signed residual is
$\delta_p = s_p(x)\cdot d$, and we use a logistic sigmoid:
\begin{equation}
\phi_{\text{closed}}\!\left(\frac{\delta_p}{\tau_p}\right)
\;=\; \frac{1}{1 + \exp\!\bigl(-2\,\delta_p/\tau_p\bigr)}
\;=\; \frac{1 + \tanh\!\bigl(\delta_p/\tau_p\bigr)}{2}.
\label{eq:kernel_closed}
\end{equation}
Interior pixels ($\delta_p \gg 0$) receive $\alpha_p \approx 1$; the
smooth $1\to 0$ transition is concentrated in a narrow band of width
$\mathcal{O}(\tau_p)$ around the boundary ($\delta_p = 0$), consistent
with the one-sided convention described in
\cref{rem:onesided_aa}. The logistic profile is $C^{\infty}$, monotone,
and symmetric about the boundary, providing stable gradients for
boundary motion in both inward and outward directions.
\section{Adaptive Subdivision Introduces Gradient Discontinuities}
\label{app:sub:adaptive}

De Casteljau subdivision splits a B\'ezier segment $\mathcal{B}(t)$,
$t\in[0,1]$, at a parameter $t^*\in(0,1)$ into two sub-segments, each
again a B\'ezier curve whose control points follow from the de Casteljau
recurrence~\cite{farin2002curves}. The adaptive variant applies this
split recursively until a flatness criterion is satisfied, yielding a
polyline whose vertex count $S(\theta)$ depends on the current
control-point configuration. In our implementation the flatness tolerance
is $0.5$\,px.

This criterion makes $S$ a piecewise-constant function of $\theta$: at
any parameter $\theta_0$ where the flatness threshold is exactly met,
$S$ jumps discontinuously. The surrogate distance
$\tilde{d}_i(x;\theta)=\mathrm{dist}\!\bigl(x,\,\mathcal{P}_{i,S(\theta)}(\theta)\bigr)$
therefore undergoes a combinatorial change in its minimizing segment
index $k^*$. By \cref{thm:danskin}, the Clarke subdifferential at such
$\theta_0$ is the convex hull of the gradients from all segments active
on either side of the transition, a set that need not contain the true
B\'ezier-distance gradient. This is the same Clarke-averaging pathology
identified for multiple closest points in
\cref{app:sub:diffvg_rooterror}, now triggered at every subdivision
boundary rather than only near medial axes.

Uniform sampling at fixed $S$ eliminates this failure mode: $S$ is
constant, no topological jump occurs, and \cref{prop:ae_diff} applies
globally. Adaptive subdivision further incurs a throughput penalty
because the per-primitive variation of $S(\theta)$ prevents fixed-stride
memory layout. Together, gradient discontinuities and implementation
overhead cause adaptive subdivision to trail uniform $S{=}24$ sampling
by $-0.06$\,dB (open) and $-0.07$\,dB (closed) at $N{=}1024$, with
lower throughput in both modes (Tab. 4 in Sec. 4.4).

\section{Low-Budget Color-Gradient Fitting}
\label{sec:low-budget}
The optimization trace in Fig.~\ref{fig:text-gradient} shows the issue is not only per-primitive color capacity but also whether the oracle allocates scarce primitives to useful regions. Overlapping coverage can efficiently express smooth gradients when optimization remains stable. The near-identical outcomes across $S$ imply that once the oracle is well-conditioned, $S$ mainly controls forward fidelity and rendering cost.

\begin{figure}[htbp]
  \centering
  \includegraphics[width=0.9\linewidth]{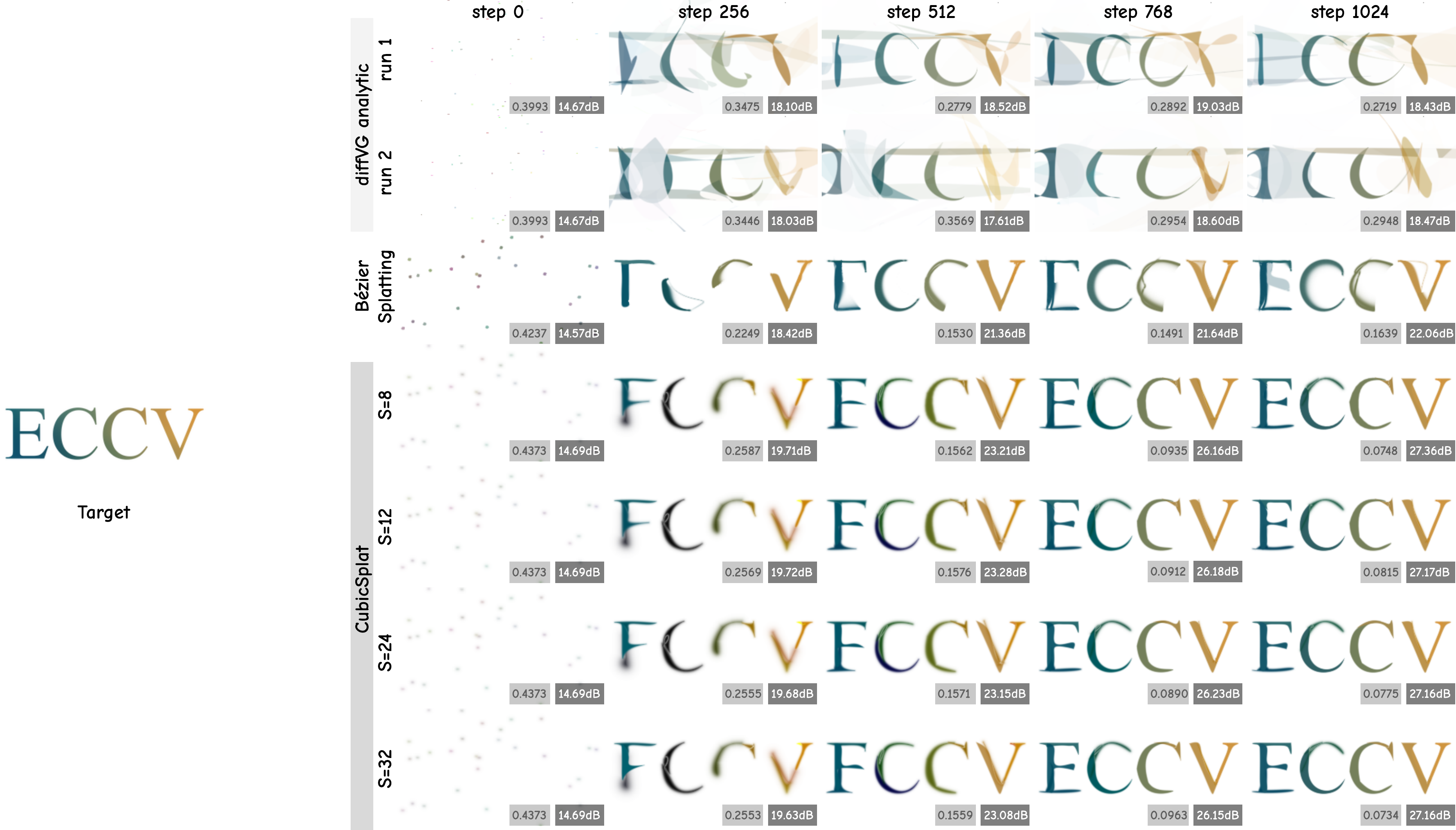}
   \caption{Low-budget gradient/text fitting with 32 closed primitives. Gray labels report LPIPS$\downarrow$ / PSNR$\uparrow$.}
   \label{fig:text-gradient}
\end{figure}

\section{Original Experiment Data}

Tables 5 and 6 in Sec. 4.5 report the per-image
PSNR and rendering time underlying the aggregated scaling analysis. Trained at 2K
($2040{\times}1344$) with 1024 curves and rendered at resolutions up to
24K ($24480{\times}16128$) on a single RTX~4090. PSNR is computed by
downsampling each rendered output back to the training resolution via
bicubic interpolation, ensuring a resolution-independent quality
comparison. Two trends are apparent from the raw numbers.
First, for every image and topology mode, PSNR varies by less than
0.5\,dB across the entire resolution range at any given $S$, confirming
the resolution invariance reported in the main text.
Second, increasing $S$ beyond 24 (open) or 12 (closed) yields
diminishing returns of at most 0.03\,dB while rendering cost grows
linearly, consistent with the $\mathcal{O}(S^{-2})$ geometric error
bound (\cref{sec:errbound}).

\label{sec:orig_data}
\input{tables/scalability}

%% file: tables/scalability.tex
\begin{table*}[t]
\centering
\caption{Resolution scalability comparison of CubicSplat. PSNR (dB) at different rendering resolutions ranging from 2K ($2040\times1344$) to 24K ($24480\times16128$) and sampling rates $K$. Left: unclosed curves; Right: closed curves. CubicSplat maintains consistent quality across all resolutions even at the lowest sampling rate ($K=8$). \textbf{Bold}: best in column; \underline{underline}: second best.}
\label{tab:resolution_scaling}
\resizebox{\linewidth}{!}{
\begin{tabular}{c|c|c|c|c|c|c||c|c|c|c|c}
\toprule
\multirow{2}{*}{\textbf{Image}} & \multirow{2}{*}{\textbf{S}} & \multicolumn{5}{c||}{\textbf{Unclosed}} & \multicolumn{5}{c}{\textbf{Closed}} \\
\cmidrule(lr){3-7}\cmidrule(lr){8-12}
 &  & \textbf{2K} & \textbf{4K} & \textbf{8K} & \textbf{16K} & \textbf{24K} & \textbf{2K} & \textbf{4K} & \textbf{8K} & \textbf{16K} & \textbf{24K} \\
\midrule
\multirow{5}{*}{0004} & 8 & 28.3370 & 28.5303 & 28.5274 & 28.5274 & 28.5273 & 29.2560 & 29.4434 & 29.4404 & 29.4404 & 29.4404 \\
 & 12 & 29.0966 & 29.3242 & 29.3207 & 29.3207 & 29.3207 & 29.2969 & 29.4829 & 29.4799 & 29.4800 & 29.4800 \\
 & 24 & \textbf{29.2937} & \textbf{29.5277} & \textbf{29.5241} & \textbf{29.5241} & \textbf{29.5241} & \textbf{29.3063} & \textbf{29.4909} & \textbf{29.4879} & \textbf{29.4880} & \textbf{29.4880} \\
 & 32 & \underline{29.2755} & \underline{29.5098} & \underline{29.5062} & \underline{29.5062} & \underline{29.5062} & \underline{29.3046} & \underline{29.4899} & \underline{29.4869} & \underline{29.4870} & \underline{29.4870} \\
 & 48 & 29.2717 & 29.5056 & 29.5020 & 29.5020 & 29.5020 & 29.3045 & 29.4899 & 29.4869 & 29.4869 & 29.4869 \\
\midrule
\multirow{5}{*}{0008} & 8 & 23.9494 & 24.0028 & 24.0040 & 24.0041 & 24.0041 & 24.5218 & 24.5781 & 24.5794 & 24.5796 & 24.5796 \\
 & 12 & 24.5520 & 24.6149 & 24.6162 & 24.6163 & 24.6163 & 24.5546 & 24.6110 & 24.6123 & 24.6124 & 24.6124 \\
 & 24 & \textbf{24.6650} & \textbf{24.7295} & \textbf{24.7308} & \textbf{24.7310} & \textbf{24.7310} & \textbf{24.5610} & \textbf{24.6171} & \textbf{24.6184} & \textbf{24.6186} & \textbf{24.6186} \\
 & 32 & \underline{24.6565} & \underline{24.7210} & \underline{24.7223} & \underline{24.7225} & \underline{24.7225} & \underline{24.5605} & \underline{24.6168} & \underline{24.6181} & \underline{24.6182} & \underline{24.6182} \\
 & 48 & 24.6543 & 24.7187 & 24.7200 & 24.7202 & 24.7202 & 24.5604 & 24.6167 & 24.6180 & 24.6182 & 24.6182 \\
\midrule
\multirow{5}{*}{0012} & 8 & 22.7299 & 22.7639 & 22.7656 & 22.7657 & 22.7656 & 23.2831 & 23.3175 & 23.3194 & 23.3196 & 23.3196 \\
 & 12 & 23.4206 & 23.4618 & 23.4637 & 23.4639 & 23.4639 & 23.3189 & 23.3534 & 23.3553 & 23.3555 & 23.3555 \\
 & 24 & \textbf{23.5555} & \textbf{23.5976} & \textbf{23.5995} & \textbf{23.5997} & \textbf{23.5997} & \textbf{23.3254} & \textbf{23.3596} & \textbf{23.3615} & \textbf{23.3617} & \textbf{23.3617} \\
 & 32 & \underline{23.5448} & \underline{23.5871} & \underline{23.5890} & \underline{23.5892} & \underline{23.5892} & \underline{23.3245} & \underline{23.3590} & \underline{23.3609} & \underline{23.3611} & \underline{23.3611} \\
 & 48 & 23.5425 & 23.5847 & 23.5866 & 23.5869 & 23.5869 & 23.3244 & 23.3589 & 23.3608 & 23.3610 & 23.3610 \\
\midrule
\multirow{5}{*}{0016} & 8 & 24.2988 & 24.2739 & 24.2774 & 24.2776 & 24.2776 & 24.5230 & 24.4942 & 24.4978 & 24.4981 & 24.4981 \\
 & 12 & 24.8432 & 24.8167 & 24.8206 & 24.8208 & 24.8209 & 24.5486 & 24.5195 & 24.5232 & 24.5234 & 24.5234 \\
 & 24 & \textbf{24.9444} & \textbf{24.9173} & \textbf{24.9212} & \textbf{24.9215} & \textbf{24.9216} & \textbf{24.5530} & \textbf{24.5237} & \textbf{24.5274} & \textbf{24.5276} & \textbf{24.5277} \\
 & 32 & \underline{24.9355} & \underline{24.9085} & \underline{24.9124} & \underline{24.9127} & \underline{24.9127} & \underline{24.5525} & \underline{24.5233} & \underline{24.5270} & \underline{24.5272} & \underline{24.5272} \\
 & 48 & 24.9335 & 24.9066 & 24.9104 & 24.9107 & 24.9108 & 24.5524 & 24.5233 & 24.5269 & 24.5272 & 24.5272 \\
\bottomrule
\end{tabular}}
\end{table*}

\begin{table*}[t]
\centering
\caption{Rendering time (ms) of CubicSplat at different resolutions and sampling rates $K$. Left: unclosed curves; Right: closed curves. CubicSplat achieves real-time rendering at standard resolutions and scales gracefully to ultra-high resolutions with moderate sampling rates.}
\label{tab:speed_scaling}
\resizebox{\linewidth}{!}{
\begin{tabular}{c|c|c|c|c|c|c||c|c|c|c|c}
\toprule
\multirow{2}{*}{\textbf{Image}} & \multirow{2}{*}{\textbf{S}} & \multicolumn{5}{c||}{\textbf{Unclosed}} & \multicolumn{5}{c}{\textbf{Closed}} \\
\cmidrule(lr){3-7}\cmidrule(lr){8-12}
 &  & \textbf{2K} & \textbf{4K} & \textbf{8K} & \textbf{16K} & \textbf{24K} & \textbf{2K} & \textbf{4K} & \textbf{8K} & \textbf{16K} & \textbf{24K} \\
\midrule
\multirow{5}{*}{0004} & 8 & \textbf{1.3366} & \textbf{3.6481} & \textbf{12.6969} & \textbf{47.5781} & \underline{144.7014} & \textbf{3.4261} & \textbf{10.3347} & \underline{94.7455} & \textbf{144.1402} & \textbf{323.6840} \\
 & 12 & \underline{1.3722} & \underline{3.8530} & \underline{13.7628} & \underline{53.6563} & \textbf{118.3065} & \underline{4.2291} & \underline{13.4328} & \textbf{49.9176} & \underline{196.6026} & \underline{684.1346} \\
 & 24 & 1.7198 & 5.0681 & 18.3116 & 72.4286 & 161.6611 & 7.2113 & 24.9134 & 95.0149 & 372.1476 & 833.3861 \\
 & 32 & 1.9629 & 5.9953 & 21.5207 & 85.8643 & 191.2358 & 9.3311 & 33.1704 & 125.7391 & 492.9564 & 1103.9075 \\
 & 48 & 2.5748 & 7.5838 & 28.4253 & 113.0126 & 252.4241 & 13.7372 & 48.9079 & 186.9312 & 733.0095 & 1639.3223 \\
\midrule
\multirow{5}{*}{0008} & 8 & \textbf{1.2711} & \textbf{3.4583} & \textbf{11.5671} & \textbf{44.3452} & \textbf{98.4073} & \textbf{3.1570} & \textbf{9.8248} & \textbf{33.6656} & \textbf{132.6116} & \textbf{364.0665} \\
 & 12 & \underline{1.3426} & \underline{3.6339} & \underline{12.6706} & \underline{49.4957} & \underline{109.7135} & \underline{4.0597} & \underline{12.9094} & \underline{48.2219} & \underline{188.1898} & 897.8553 \\
 & 24 & 1.9780 & 4.9855 & 16.8722 & 66.5320 & 146.5880 & 7.0901 & 24.0546 & 90.9575 & 355.0058 & \underline{793.3235} \\
 & 32 & 1.8537 & 5.5037 & 19.5870 & 78.2072 & 173.4529 & 8.9947 & 31.8002 & 120.0111 & 468.4794 & 1068.8938 \\
 & 48 & 2.2789 & 7.0046 & 25.5493 & 102.9074 & 229.1957 & 13.2606 & 46.9189 & 178.8374 & 696.2420 & 1555.4501 \\
\midrule
\multirow{5}{*}{0012} & 8 & 5.9444 & \textbf{3.6320} & \textbf{12.7378} & 71.7889 & \underline{144.3720} & \textbf{3.4497} & \textbf{11.0056} & \underline{79.9288} & \textbf{179.8782} & \textbf{369.1505} \\
 & 12 & \textbf{1.3747} & \underline{3.8357} & \underline{13.5214} & \textbf{52.8268} & \textbf{117.5839} & \underline{4.4052} & \underline{14.6077} & \textbf{54.8447} & \underline{215.5997} & \underline{597.5115} \\
 & 24 & \underline{1.6989} & 5.0161 & 17.9776 & \underline{71.0831} & 159.0746 & 7.8738 & 27.5331 & 105.2806 & 411.9568 & 919.6869 \\
 & 32 & 1.9429 & 5.8492 & 20.9490 & 83.7392 & 188.0708 & 10.2560 & 36.5170 & 138.9592 & 544.3014 & 1217.3083 \\
 & 48 & 2.3907 & 7.4063 & 27.7210 & 111.2797 & 245.9742 & 14.8925 & 54.7229 & 208.7724 & 809.3464 & 1833.2964 \\
\midrule
\multirow{5}{*}{0016} & 8 & \underline{1.4810} & \textbf{3.4984} & \textbf{11.7264} & \textbf{43.5104} & \textbf{95.6209} & \textbf{3.3538} & \textbf{10.6105} & \textbf{42.4438} & \textbf{150.5988} & \textbf{545.2135} \\
 & 12 & \textbf{1.2937} & \underline{3.6522} & \underline{12.2768} & \underline{48.2425} & \underline{106.8715} & \underline{4.7180} & \underline{14.5114} & \underline{52.1277} & \underline{204.8995} & \underline{837.3899} \\
 & 24 & 1.5831 & 4.5279 & 16.1065 & 63.7861 & 142.2179 & 7.6296 & 26.0071 & 99.5169 & 389.8328 & 872.8230 \\
 & 32 & 1.7906 & 5.2216 & 18.8405 & 74.9519 & 168.0772 & 9.7809 & 34.6942 & 131.6302 & 507.0663 & 1140.5719 \\
 & 48 & 2.1831 & 6.8385 & 24.5286 & 98.0961 & 219.7528 & 13.8946 & 50.0666 & 192.2672 & 752.7704 & 1683.4849 \\
\bottomrule
\end{tabular}}
\end{table*}